\documentclass{article}
\usepackage[utf8]{inputenc}
\usepackage{amsmath}
\usepackage{amssymb}
\usepackage{graphicx}
\usepackage[dvipsnames]{xcolor}

\usepackage{color-edits}

\usepackage{graphicx}
\addauthor[Kiarash]{kb}{cyan}
\addauthor[pasha]{pasha}{red}
\addauthor[Danny]{dm}{Green}
\addauthor[Mahdi]{m}{purple}
\usepackage{amsthm}
\usepackage[margin=1in]{geometry}
\usepackage{tikz}
\usepackage[tight]{subfigure}
\usepackage[hypertexnames=false]{hyperref}
\usepackage{cleveref}

\newcommand{\country}[1]{#1.}
\newcommand{\city}[1]{#1}
\newcommand{\institution}[1]{#1}
\newcommand{\email}[1]{Email: \texttt{#1}}
\newcommand{\affiliation}{\thanks}
\author{
 {Kiarash Banihashem
 \affiliation{
   \institution{University of Maryland}
   \city{College Park}
   \country{USA}
 \email{kiarash@umd.edu}
 }}
 \and
 {MohammadTaghi Hajiaghayi
 \affiliation{
   \institution{University of Maryland}
   \city{College Park}
   \country{USA}
 \email{hajiaghayi@gmail.com}
 }}
 \and
 {Mahdi JafariRaviz
 \affiliation{
   \institution{University of Maryland}
   \city{College Park}
   \country{USA}
 \email{mahdijafariraviz@gmail.com}
 }}
 \and
 {Danny Mittal
 \affiliation{
   \institution{University of Maryland}
   \city{College Park}
   \country{USA}
 \email{dannymittal@gmail.com}
 }}
 \and
 {Alipasha Montaseri
 \affiliation{
   \institution{University of Maryland}
   \city{College Park}
   \country{USA}
 \email{alipasha.montaseri@gmail.com}
 }}
}

\newcommand{\tempdel}[1]{}

\makeatletter
\@addtoreset{subsubsection}{subsection}
\makeatother

\newcommand{\R}{\mathbb{R}}
\newcommand{\Z}{\mathbb{Z}}

\newcommand{\EqComment}[1]{\text{\emph{(#1)}}}

\newcommand{\pneg}{{p_{\text{neg}}}}
\newcommand{\pzero}{{p_{\text{zero}}}}
\newcommand{\ppos}{{p_{\text{pos}}}}

\newcommand{\inv}{^{-1}}

\newtheorem{definition}{Definition}
\newtheorem{theorem}{Theorem}
\newtheorem{lemma}{Lemma}
\newtheorem{corollary}{Corollary}

\newtheorem{example}{Example}

\newcommand{\norm}[1]{{ \Vert #1 \Vert}}

\newcommand{\munagalamodel}{scalar general-cost } 
\newcommand{\munagalamodelcapital}{Scalar General-Cost }

\newcommand{\vectorfjmodel}{multidimensional FJ }
\newcommand{\vectorfjmodelcapital}{Multidimensional FJ }

\newcommand{\fjmodel}{FJ }

\newcommand{\vectormunagalamodel}{general-cost multidimensional }
\newcommand{\vectormunagalamodelcapital}{General-Cost Multidimensional }

\newcommand{\cliquevectormodel}{multidimensional clique }
\newcommand{\cliquevectormodelcapital}{Multidimensional Clique }

\newcommand{\quadraticvectormodel}{multidimensional quadratic }
\newcommand{\quadraticvectormodelcapital}{Multidimensional Quadratic }

\newcommand{\ourmodel}{multidimensional heterogeneous }
\newcommand{\ourmodelcapital}{Multidimensional Heterogeneous }

\newcommand{\arbitrarysymmetricmodel}{arbitrary symmetric cost }
\newcommand{\arbitrarysymmetricmodelcapital}{Arbitrary Symmetric Cost }

\newenvironment{theoremrestate}[1]{
  \begingroup
  \renewcommand{\thetheorem}{#1}
  \begin{theorem}
}{
  \end{theorem}
  \addtocounter{theorem}{-1}
  \endgroup
}

\newenvironment{lemmarestate}[1]{
  \begingroup
  \renewcommand{\thelemma}{#1}
  \begin{lemma}
}{
  \end{lemma}
  \addtocounter{lemma}{-1}
  \endgroup
}

\usepackage{csquotes}

\title{How Bad Is Forming Your Own Multidimensional Opinion?}
\date{}

\begin{document}

\maketitle

\begin{abstract}
    Understanding the formation of opinions on multiple interconnected topics within social networks is of significant importance. It offers insights into collective behavior and decision-making processes, with applications in Graph Neural Networks. Existing models propose that individuals form opinions based on a weighted average of their peers' opinions and potentially their own beliefs. This averaging process, when viewed as a best-response game, can be seen as an individual minimizing disagreements with peers, defined by a quadratic penalty, leading to an equilibrium. Bindel, Kleinberg, and Oren (FOCS 2011) provided tight bounds on the \enquote{price of anarchy,} which is defined as the maximum level of overall disagreement at equilibrium relative to a social optimum. Bhawalkar, Gollapudi, and Munagala (STOC 2013) generalized the penalty function to consider non-quadratic penalties and provided tight bounds on the price of anarchy of these functions.

    When considering multiple topics, an individual's opinions can be represented as a vector. Parsegov, Proskurnikov, Tempo, and Friedkin (2016) proposed a multidimensional model using the weighted averaging process, but with constant interdependencies between topics for any individual. However, the same question of the price of anarchy for this model remained open. We address this question by providing tight bounds on the price of anarchy of the multidimensional model, while also generalizing the model to consider more complex interdependencies. Furthermore, through novel approaches, following the work of Bhawalkar, Gollapudi, and Munagala, we provide tight bounds on the price of anarchy under non-quadratic penalty functions. Surprisingly, these bounds match the bounds for the scalar model used by them. We further demonstrate that the bounds for the price of anarchy remain unchanged even when we add another layer of complexity to the dynamics, involving multiple individuals working in groups to minimize their overall internal and external disagreement penalty, a common occurrence in real-life scenarios. Lastly, we provide a novel lower bound showing that for any reasonable penalty function, the worst-case price of anarchy across all networks is at least $\frac 2 {e\ln 2}$. This lower bound naturally applies to both the multidimensional and scalar models and was previously unknown even for scalar models.
\end{abstract}

\section{Introduction}

The behaviors and choices of individuals are frequently shaped by those around them. Social dynamics and peer influences play a crucial role in forming an individual's preferences and decisions. For instance, if most of a person's friends frequently visit a particular restaurant, that person is likely to develop a preference for dining there as well. Understanding these patterns of influence is essential for comprehending various social phenomena, from cultural trends to market behaviors.

Various models have been proposed \cite{degroot1974reaching, friedkin1990social, bhawalkar2013coevolutionary, parsegov2016novel, ravazzi2014ergodic} to examine the dynamics of opinion formation under different constraints. In most of these models, a \emph{weighted} (directed or undirected) graph $G = (V, E)$ represents the social network, where each node $i$ has an \emph{expressed} opinion $z_i$. We sometimes refer to this expressed opinion simply as an opinion. This opinion can be either a single number (representing an opinion on a single topic) or a vector (representing opinions on multiple topics). Every edge $(i, j) \in E$ signifies the influence of node $j$ on node $i$'s opinion (or both ways if the edge is undirected), and the weight of the edge quantifies this influence. Opinions evolve through a coevolutionary process: at each time step, node $i$ updates its opinion based on the opinions of its neighbors in $G$.

In most models, node $i$'s opinion $z_i$ is a single scalar, and in most of the rest (where $z_i$ is a vector), the different entries of $z_i$ are independent of each other (i.e. the topics are independent). However, in reality, people have opinions on a wide variety of topics, and when these topics are related to each other, one person's opinion on one topic can affect another person's opinion on a different topic. Therefore, it is of interest to consider models of multidimensional opinions with interdependencies.
Parsegov, Proskurnikov, Tempo, and Friedkin \cite{parsegov2016novel} introduced a multidimensional model that captures these dependencies between topics. To the best of our knowledge, the model introduced by \tempdel{Parsegov et al.\ }\cite{parsegov2016novel} is the only one in the literature that considers vector opinions with dependencies between topics.

An important question in these models is the extent of overall disagreement that would persist if individuals' opinions were allowed to converge naturally, compared to the optimal scenario where societal disagreement is minimized. For instance, if each person’s opinion represents their position on a political spectrum, it becomes important to understand how much political disagreement would remain relative to the minimum possible level of disagreement. In mathematical terms, this comparison corresponds to the price of anarchy, which captures the inefficiency that arises when individuals act based on their own interests rather than toward a collective goal. Providing bounds on the price of anarchy is essential to understanding how far the outcome deviates from the optimal case, highlighting the potential gap between natural opinion convergence and what could be achieved through more structured, cooperative efforts.

The dynamics of opinion formation have applications in Graph Neural Networks (GNNs) \cite{scarselli2008graph, gori2005new, frasconi1998general, sperduti1997supervised}. Specifically, \tempdel{Scarselli et al.\ }\cite{scarselli2008graph, gori2005new} proposed a popular model for processing data represented in graph domains. In their model, each node $i$ in the graph has a label $l_i$ and an associated state vector $x_i$ of size $m$. These states are dynamically updated based on the relationships (edges) among the nodes, which are weighted using learnable parameters $w$. Finally, an output is computed using the states of the nodes. A learning algorithm can then be created using gradient descent on the parameters $w$. When the dynamic updates are performed linearly, the connection to opinion formation dynamics becomes apparent. However, it is important to note that in this model, the states attached to the nodes are vectors (of size $m$) and thus the study of multidimensional opinion formation is significant.

The following paragraphs provide an overview of several opinion formation models and their relevant results. For convenience, we have given each model a name to simplify their reference.

\paragraph{The DeGroot model.}
The \emph{DeGroot} model \cite{degroot1974reaching}, is a basic model of opinion formation, in which each person $i$ holds an opinion represented by a real number $z_i$. The network is modeled by a weighted graph $G = (V, E)$ with $n$ nodes, where each node is influenced by its neighbors and the edge weight represents the extent of the influence. Then in each step, each node changes its opinion to a weighted average of its neighbors. More formally, let $z_i(t)$ denote the opinion of node $i$ at time step $t$. Then if $w_{ij} \ge 0$ denotes the weight on the edge $(i, j)$, the updates are defined as:
\begin{align}
    z_i(t + 1) = \frac{\sum_{j \in N(i)}w_{ij}z_j(t)}{\sum_{j \in N(i)} w_{ij}},
\end{align}
for every node $i$, where $N(i)$ is the set of neighbors of $i$ in $G$. Under general conditions, this process leads to a state of \textit{consensus}, meaning that every node holds the same opinion. Therefore, a different model is needed to account for situations where consensus is not reached.

\paragraph{The \fjmodel model.}
Friedkin and Johnsen \cite{friedkin1990social} proposed a variation of the DeGroot model  (which we will refer to as the \emph{\fjmodel} model), where each node holds an internal opinion $s_i$ that remains constant even as the node's expressed opinion $z_i$ changes. Furthermore, $r_i \ge 0$ represents the weight that node $i$ attaches to its internal opinion. Then the updates are defined as:
\begin{align}
    z_i(t + 1) = \frac{r_is_i + \sum_{j \in N(i)} w_{ij}z_j(t)}{r_i + \sum_{j \in N(i)} w_{ij}},
\end{align}
for every node $i$. One can observe that when $r_i = 0$ for every $i$, the \fjmodel model is equivalent to the DeGroot model. Bindel, Kleinberg, and Oren \cite[FOCS 2011]{bindel2015bad} interpreted the updates in the above equation as choosing $z_i$ to minimize the following cost function:
\begin{align}
    \label{eq:fj:cost}
    c_i(z) = r_i(z_i - s_i)^2 + \sum_{j \in N(i)} w_{ij} (z_i - z_j)^2,
\end{align}
where $z = (z_1, \dots, z_n)$ is the sequence of all opinions. From this perspective, the repeated averaging can be thought of as the trajectory of best-response dynamics in a one-shot, complete information game played by the nodes in $V$, where node $i$'s strategy is a choice of opinion $z_i$ and its payoff is the negative of the cost in \Cref{eq:fj:cost}. Note that when $r_i > 0$ for every $i$, the repeated averaging process always converges to a fixed point, which is the unique Nash equilibrium of the game.

Defining the social cost $SC(z)$ as the sum of each node's cost, i.e. $SC(z) = \sum_i c_i(z)$, one can observe that the Nash equilibrium vector $x = \begin{pmatrix} x_1, \dots, x_n\end{pmatrix}$ does not always correspond to the social optimum vector, which is the vector $y$ minimizing $SC(y)$. A natural question is thus the price of anarchy (PoA), defined as the ratio between the cost of the Nash equilibrium and the cost of the optimal solution (i.e. $\frac{SC(x)}{SC(y)}$). \tempdel{Bindel et al.\ }\cite{bindel2015bad} have shown that when the graph is undirected (i.e. $w_{ij} = w_{ji}$ for every $i,j$), the PoA is at most $\frac{9}{8}$. They also demonstrate that the PoA can be unbounded for directed graphs.

\paragraph{The \munagalamodelcapital model.} 
\tempdel{Bhawalkar et al.\ } \cite{bhawalkar2013coevolutionary} generalize the cost function so that the node $i$'s cost function is defined as:
\begin{align}
    \label{bhawalkar:cost}
    c_i(z) = g_i(z_i - s_i) + \sum_{j \neq i} f_{ij}(z_i - z_j),
\end{align}
where $f_{ij}$ and $g_i$ are real valued functions that remain fixed. They generally require that $f_{ij}$ and $g_i$ be symmetric (in the sense that $f(d) = f(-d)$) and convex in order to obtain their results. One can observe that the cost function introduced by \tempdel{Bindel et al.\ }\cite{bindel2015bad} is a specific case of this model where $f_{ij}(x) = w_{ij}x^2$ and $g_i(x) = r_i x^2$.

\paragraph{The \vectorfjmodelcapital model.} Another line of work is considering nodes' opinions on multiple topics. \tempdel{Parsegov et al.\ }\cite{parsegov2016novel} introduced a multidimensional model, that considers $m$ different topics, meaning each person holds $m$ opinions (i.e. $z_i \in \R^m$ is a vector). This model accounts for the fact that a person's opinion on one topic might be influenced by others' opinions on a different topic, implying interdependencies between topics.
More specifically the expressed opinion and the internal opinion of node $i$, denoted as $z_i$ and $s_i$ respectively, are real vectors of size $m$. A round of updates is then defined as:
\begin{align}
    \label{eq:vectorized-fj:update}
    z_i(t + 1) = \frac{r_is_i + C \sum_{j \in N(i)} w_{ij} z_j(t)}{r_i + \sum_j w_{ij}},
\end{align}
where $C \in \R^{m \times m}$ is the matrix of multi-issue dependence structure (MiDS). Matrix $A$ is \emph{row-stochastic} if $A_{ij} \ge 0$ and $\sum_j A_{ij} = 1\ \forall i$. As noted by \tempdel{Parsegov et al.}\cite{parsegov2016novel}, choosing $C$ to be row-stochastic is natural, as it makes the update correspond to a \enquote{weighted average}. We call an instance of this model \emph{undirected} if $w_{ij} = w_{ji}\ \forall i,j$ and $C$ is symmetric. Note that the symmetry of $C$ is important, as an asymmetric $C$, intuitively, introduces asymmetry in the way two nodes connected with an edge influence each other.

\section{Our Contributions}
\label{section:our-contributions}

\subsection{The \ourmodelcapital Model}

In this paper, we introduce the following model, which we refer to as the \textit{\ourmodel model}. All of our results are attained in this model; we describe other models only to show how they are generalized by or can be reduced to the \ourmodel model. In this model, we use the following cost function:
\begin{align}
    c_i(z) = g_i(R_iz_i - s_i) + \sum_{j \neq i} f_{ij}(A_{ij}z_i + B_{ij}z_j),
\end{align}
where for each person $i$ we have a dimension $m_i \in \Z$ such that $z_i \in \R^{m_i}$. We then have a dimension $d_{ij}$ so that $f_{ij} : \R^{d_{ij}} \to \R$, and $A_{ij}, B_{ij}$ are linear transformations $A_{ij} : \R^{m_i} \to \R^{d_{ij}}$, $B_{ij} : \R^{m_j} \to \R^{d_{ij}}$. Similarly, we have a dimension $e_i$ so that $g_i : \R^{e_i} \to \R$, $R_i$ is a linear transformation $\R_i : \R^{m_i} \to \R^{e_i}$, and $s_i \in \R^{e_i}$ is the equivalent of an internal opinion.

Since the model arises from a game interpretation, in the following sections, we often refer to an instance of it as an \emph{\textbf{opinion formation game}}. We call an opinion formation game, \emph{\textbf{symmetric}}, when $d_{ij} = d_{ji}$, $f_{ij} = f_{ji}$, $A_{ij} = B_{ji}$, $B_{ij} = A_{ji}$ for all $i,j$, so that for each pair of persons $(i, j)$ their costs due to the other person are the same (i.e. $f_{ij}(A_{ij}z_i + B_{ij}z_j) = f_{ji}(A_{ji} z_j + B_{ji} z_i)$). Note that some previous works refer to an undirected network; the definition of a symmetric opinion formation game captures the same meaning. All of our results are for symmetric opinion formation games; as shown by \tempdel{Bindel et al.\ }\cite{bindel2015bad}, even in their much simpler model, allowing asymmetry can lead to the PoA being unbounded.

Note that in comparison to the \munagalamodel model defined in \Cref{bhawalkar:cost}, the \ourmodel model does not generally demand that $f_{ij}$ themselves be symmetric in the sense that $f(v) = f(-v)$, because this symmetry is now baked in to $A_{ij}, B_{ij}$. As a result, our model provides more flexibility even in the case where all cost functions are $\R \to \R$ and all matrices $A_{ij}, B_{ij}, R_{ij}$ are scalars. For example, we could define the penalty between two people $1, 2$ to be $e^{z_1 - z_2}$ by setting $f_{12}(x) = e^x$, $A_{12} = B_{21} = 1$, and $B_{12} = A_{21} = -1$; though $e^x$ is not itself symmetric, the cost for each player remains the same due to our definition of $A_{ij}, B_{ij}$. Our PoA bounds then extend to this kind of example, which was previously not considered under models of symmetric opinion formation games.

The \ourmodel model can be interpreted as allowing different persons to have not only different opinions but different frameworks in which they form opinions, so that different persons' opinions may have different structures. We then use the matrices $A_{ij}, B_{ij}$ to allow one aspect of person $i$'s opinion to interact with another aspect of person $j$'s opinion, instead of forcing all persons' opinions to be structured the same way. For this reason, we refer to the model as \textit{heterogenous}; this is in addition to the model being \textit{multidimensional} because the opinions are not necessarily scalars.

In \Cref{section:equiv-models}, we discuss several models that are special cases of the \ourmodel model. These models either directly extend previous work or introduce new dynamics. First, we introduce the \emph{\quadraticvectormodel}model, which generalizes the model studied by Parsegov et al. \cite{parsegov2016novel}. Next, we present the \emph{\vectormunagalamodel}model, a multidimensional version of the model introduced by Bhawalkar et al. \cite{bhawalkar2013coevolutionary}. We then examine the \emph{\arbitrarysymmetricmodel}model, which is chosen to clearly allow for the most general possible symmetric costs in networks. We show that even this general model is equivalent to the \ourmodel model, showing that the \ourmodel model is also the most general possible model of symmetric cost functions in networks. Finally, we introduce the \emph{\cliquevectormodel}model, which follows different dynamics.

\subsection{General Upper Bounds on the Price of Anarchy}
We show a general upper bound on the PoA of a symmetric opinion formation game in section \ref{section:general-ub}. The value of this upper bound depends on a property of the cost functions $f_{ij}, g_i$ that we will term $(\lambda, \kappa, p)$-\textit{suitability}.
\begin{definition}[$(\lambda, \kappa, p)$-suitability]
    \label{suitable}
    A differentiable function $f : \R^m \to \R$ is said to be $(\lambda, \kappa, p)$-\textbf{suitable} if for all $a, b \in \R^m$, we have
    \begin{align}
    \label{ineq:suitability}
    (\nabla f(a))^T \frac {b - a} p \leq \lambda f(b) - \kappa f(a).
    \end{align}
\end{definition}

We then show how $(\lambda, \kappa, p)$-suitability relates to the analysis of the PoA.
\begin{theorem}\label{suitabilitythm}
    Consider a symmetric opinion formation game whose cost functions $f_{ij}, g_i$ are nonnegative and differentiable. If for some $\lambda, \kappa > 0$, we have that all $f_{ij}$ are $(\lambda, \kappa, 2)$-suitable and all $g_i$ are $(\lambda, \kappa, 1)$-suitable, then the price of anarchy of the opinion formation game is at most $\frac \lambda \kappa$.
\end{theorem}
Note in particular that $f_{ij}, g_i$ are not required to be convex for this result to hold. However, most of our following results will require them to be convex. In particular, we can show that if the cost functions are allowed to be non-convex then the price of anarchy can become unbounded, even in the single-dimensional case (see Appendix \ref{section:non-convex_unbounded_example}). Note that we generally take the Nash equilibrium $x$ to be defined by satisfying $\nabla_i c_i(x) = 0$, but this is only valid when the cost functions involved are convex.

We then define the \emph{worst-case} PoA, which helps to consider how high the price of anarchy can be when the cost functions are of a specific form.

\begin{definition}[Worst-Case Price of Anarchy]
    \label{definition:worst-case-poa}
    The \emph{worst-case price of anarchy} \emph{(worst-case PoA)} of a function $h : \R^m \to \R$ is the supremum of the PoA across all symmetric opinion formation games whose cost functions are of the form $f_{ij}(x) = w_{ij}h(x)$ and $g_i(x) = r_ih(x)$, where $w_{ij}, r_i$ are nonnegative real numbers.
\end{definition}

The main utility of \Cref{suitabilitythm} is to show upper bounds on the worst-case PoAs of various functions; this use is expressed in the following corollary:
\begin{corollary}
    \label{suitcorol}
    Suppose that a function $h : \R^m \to \R$ is differentiable. If for some $\lambda, \kappa > 0$, $h$ is both $(\lambda, \kappa, 1)$-suitable and $(\lambda, \kappa, 2)$-suitable, then we will say that $h$ is simply $(\lambda, \kappa)$-\textbf{suitable}, and the worst-case PoA of $h$ is bounded by $\frac \lambda \kappa$.
\end{corollary}
\begin{proof}
    For any network considered in the definition of the worst-case PoA of $h$, its cost functions $f_{ij}$ are all of the form $w_{ij}h$, so because $h$ is $(\lambda, \kappa, 2)$-suitable, all $f_{ij}$ are also $(\lambda, \kappa, 2)$-suitable. Similarly, its cost functions $g_i$ are all of the form $r_i h$, so because $h$ is $(\lambda, \kappa, 1)$-suitable, all $g_i$ are also $(\lambda, \kappa, 1)$-suitable. (These implications are special cases of \Cref{suitlinear}, which is proven below). Therefore, by \Cref{suitabilitythm}, the PoA of any such network is at most $\frac \lambda \kappa$ so the worst-case PoA of $h$ is at most $\frac \lambda \kappa$.
\end{proof}

\subsection{Bounds on Worst-Case Price of Anarchy for Specific Functions}
We use our framework to provide bounds on the worst-case price of anarchy (as defined in \Cref{definition:worst-case-poa}) for specific functions. The following theorem summarizes our results.

\begin{theorem}
    \label{theorem:x-alpha-norm}
    For $h : \R^m \to \R$ defined as $h(z) = \norm{z}^\alpha$, where $\norm{\cdot}$ is any nonconstant norm on $\R^m$ which is differentiable everywhere (except possibly at the origin), the worst-case PoA is exactly:
    \begin{align} \label{eq:zetaintro}
        \zeta(\alpha) = \frac {(\alpha - 1)^{\alpha - 1}} {\alpha^\alpha} \cdot \frac {\left(2^{\frac \alpha {\alpha - 1}} - 1\right)^\alpha} {2^{\frac \alpha {\alpha - 1}} - 2}.
    \end{align}
\end{theorem}
The above theorem is a generalization of the results from \tempdel{Bhawalkar et al.\ }\cite[Lemma 3.10]{bhawalkar2013coevolutionary}, which show this bound when $m = 1$ (the scalar game). They also show that the bound is tight for $m = 1$, which naturally extends to the multidimensional game.

The idea behind the proof of the above theorem \ref{theorem:x-alpha-norm} is that we show for the cost function $f \circ g$, where $g$ is vector-to-scalar and $f$ is scalar-to-scalar, the $(\lambda, \kappa, p)$-suitability of $f$ is preserved when $g$ is chosen from a certain class of functions. An example of a class of functions for which we show this is the class of nonconstant and differentiable norms; this is applied in the theorem above.

\subsection{A Lower Bound on Price of Anarchy}

We show the following theorem, which presents a novel lower bound for the worst-case PoA that was previously unknown even in the \munagalamodel model of \tempdel{Bhawalkar et al.\ }\cite{bhawalkar2013coevolutionary}.
\begin{theorem}
    \label{theorem:lowerbound}
    \label{lowerboundthm}
    For any nonzero function $h : \R^m \to \R$ which is nonnegative, differentiable, and convex, the worst-case PoA of $h$ is at least $\frac 2 {e \ln 2}$.
\end{theorem}

The above theorem means that for any such $h$, a symmetric game exists with cost functions which are all $h$ scaled by some nonnegative weight where the PoA is greater than or arbitrarily close to $\frac 2 {e \ln 2}$.

This result is significant and somewhat surprising for the following reason. Based on \Cref{theorem:x-alpha-norm}, we see that specifically for functions of the form $\norm{z}^\alpha$, cost functions with higher growth rates lead to smaller bounds on the PoA. Given this, we might expect that as we increase the growth rate of the cost function, perhaps even choosing the growth rate to be faster than polynomial, the PoA approaches $1$. However, \Cref{theorem:lowerbound} tells us that this is not the case: regardless of the cost functions used, it is always possible to choose a network so that the PoA approaches $\frac 2 {e \ln 2}$.

Interestingly, if we consider the worst-case PoA of $\norm{z}^\alpha$ as $\alpha$ goes to infinity, this bound indeed approaches $\frac 2 {e \ln 2}$, as illustrated in Figure \ref{fig:boundalpha}. This can be interpreted as follows: to get arbitrarily close to minimizing the worst-case PoA, it is sufficient to choose polynomial cost functions of sufficiently high degree. That is, we do not need to consider functions that grow faster than polynomials in order to keep the worst-case PoA small.

The idea behind the proof of \Cref{theorem:lowerbound} comes from considering the worst-case PoA of $e^x$ -- the suitability inequality (Inequality \ref{ineq:suitability}) for $e^x$ can be simplified by defining $c = b - a$ and replacing it in the inequality, so that the inequality is only in terms of $c$, after which it can be optimized using basic calculus. We then crucially apply this same idea to general functions by writing $f(x) = e^{g(x)}$ and defining $c = g(b) - g(a)$ and replacing it in the inequality, arguing that despite the more general form of $g$, we can derive the same inequality that is immediate in the case of $g(x) = x$.

In addition to the aforementioned consequences, \Cref{theorem:lowerbound} also allows us to determine the exact worst-case PoA of more exotic cost functions for which tight PoA bounds were previously unknown. For example, by combining \Cref{theorem:lowerbound} with \Cref{suitabilitythm}, we can show that the worst-case PoA of $\cosh x$ is exactly $\frac 2 {e \ln 2}$. These results are discussed in \Cref{subsection:exact-poa}.

\begin{figure}[t!]
    \centering
    \includegraphics[width=0.5\textwidth]{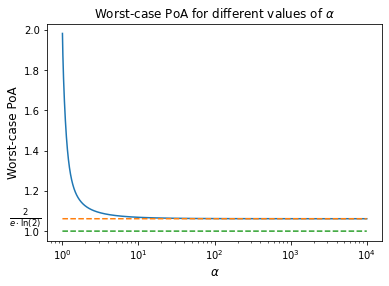} 
    \caption{The worst-case PoA of $\norm{z}^\alpha$ (Equation \ref{eq:zetaintro}) for different values of $\alpha$  is shown by the blue curve. As $\alpha$ approaches $\infty$, it converges to $\frac{2}{e \ln{2}}$.}
    \label{fig:boundalpha}
\end{figure}

\subsection{Special Cases and Equivalent Models}
\label{section:equiv-models}

In this section, we explore three key special cases of the \ourmodel model and compare them with the existing models discussed in the literature. We additionally list as a fourth model a new model that we introduce which allows for opinions to be formed via more general group dynamics so that the definition of PoA changes; while not a special case of the \ourmodel model, it can in fact be reduced to the latter. Thus, our results for the \ourmodel model extend to all four of these models in addition to all previously discussed models.

\textbf{The \quadraticvectormodelcapital Model.} This model generalizes the \vectorfjmodel model. First, we assign a matrix $W_{ij} \in \R^{m \times m}$ as the \emph{weight} of edge $(i, j)$ (instead of the scalar weight $w_{ij}$). For simplicity, assume that $W_{ij} = 0$ if $(i, j) \notin E$. Furthermore, we denote a matrix $R_i \in \R^{m \times m}$ for the \emph{weight} that node $i$ assigns to its internal opinion. We then introduce a cost function for each node $i$ as:
\begin{align}
    \label{eq:quadratic-vectorized:cost}
    c_i(z) = (z_i - s_i)^T R_i (z_i - s_i) + \sum_j (z_i - z_j)^T W_{ij} (z_i - z_j),
\end{align}
where $z = (z_1, \dots, z_n)$ is the sequence of all opinions. We call an instance of this model \emph{undirected} if $W_{ij} = W_{ji}$ for every $i,j$. We mainly study undirected instances of this model in the following sections.

We assume that $W_{ij}$ is symmetric, as replacing it with $\frac{1}{2} (W_{ij} + W_{ij}^T)$ does not change $c_i(z)$. The same argument holds for $R_i$. We require that $W_{ij}$ is positive semidefinite, so that $(z_i - z_j)^T W_{ij} (z_i - z_j)$ is minimized at $z_i = z_j$. Similarly, the same argument holds for $R_i$, but we mostly further assume that $R_i$ is positive definite for update rules to be well defined (see \Cref{eq:quadratic-vectorized:update}). Positive semidefinitness of $W_{ij}$ and $R_i$ is a reasonable assumption, as without it, a Nash equilibrium may not exist (see Appendix \ref{section:non_positive_definite}).

To minimize $c_i(z)$, the following update process can be used:
\begin{align}
    \label{eq:quadratic-vectorized:update}
    z_i(t + 1) = \left( R_i + \sum_j W_{ij} \right)\inv \left(R_i s_i + \sum_j W_{ij} z_j(t) \right),
\end{align}
which is valid under the assumption that $R_i$ is positive definite, because then $R_i + \sum_j W_{ij}$ is positive definite and thus invertible.

In the \vectorfjmodel model, the MiDS (i.e. the matrix $C$) is constant (up to a scalar factor) for every $i$ and $j$, but in this model, we have removed this constraint. Specifically, consider an instance of the \vectorfjmodel model with $n$ nodes, with internal weights $(r_i)_{i=1}^n$, internal opinions $(s_i)_{i=1}^n$, edge weights $(w_{ij})_{i,j=1}^n$, and $C$ as the MiDS. When the graph has no self-loops, and $C$ is symmetric and row-stochastic, and $r_i > 0$ for all $i$, defining $W_{ij} = w_{ij}C$, and $R_i = (r_i + \sum_j w_{ij})I - \sum_j w_{ij} C$ and using $R_i\inv r_i s_i$ as the internal opinion for node $i$, we get an equivalent instance in the \quadraticvectormodel model. Note that the new instance is valid because $r_i > 0$ results in a positive definite $R_i$, and $w_{ij} C$ is positive semidefinite since $w_{ij} \ge 0$ and $C$ is positive semidefinite under these settings. One can rewrite the update rule in \Cref{eq:quadratic-vectorized:update} for the defined instance to get an equivalent update rule to the \vectorfjmodel update rule in \Cref{eq:vectorized-fj:update}.

For completeness, we discuss the convergence of the update rule in \Cref{eq:quadratic-vectorized:update} in the \Cref{appendix:convergence}. In short, we show that for any undirected instance having positive definite $R_i$ for all $i$, we can apply a modification to the network that makes the instance convergent regardless of the starting point (i.e. $z_i(0)$ for all $i$) (see \Cref{theorem:quadratic-vectorized:convergence}). This modification is shown to maintain the Nash equilibrium (i.e. the convergence point).

\textbf{The \vectormunagalamodelcapital  Model.}
This model generalizes the model of \tempdel{Bhawalkar et al.\ }\cite{bhawalkar2013coevolutionary} to allow for multidimensional opinions. The cost function is as follows:
\begin{align}
    \label{equation:munagala-cost}
    c_i(z) = g_i(z_i - s_i) + \sum_{j \neq i} f_{ij}(z_i - z_j),
\end{align}
for $z_i \in \R^m$, where $f_{ij}$ and $g_i$ are real-valued functions that remain fixed. These functions take a vector of size $m$ as input and output a real number. When $m = 1$, this model reduces to the \munagalamodel model, which itself reduces to the \fjmodel model for quadratic functions $f_{ij}$ and $g_i$. Note that when $m=1$ (i.e.\ the opinions are scalars),
this model is equivalent to the \munagalamodel model \ref{bhawalkar:cost}.

To derive the \quadraticvectormodel model introduced in \Cref{eq:quadratic-vectorized:cost}, we observe that $x^T A x$ can be rewritten as $\norm{A^\frac{1}{2} x}^2$ for positive semidefinite $A$, where $\norm{\cdot}$ denotes the Euclidean norm. Therefore, setting $f_{ij}(x) = \norm{W_{ij}^\frac{1}{2} x}^2$ and $g_i(x) = \norm{R_i^\frac{1}{2} x}^2$ yields the \quadraticvectormodel model (\ref{eq:quadratic-vectorized:cost}).

\textbf{The \arbitrarysymmetricmodelcapital Model.} This model provides the most general possible symmetric costs in a network, where symmetry is taken to imply that the cost to person $i$ due to person $j$ is the same as the cost to person $j$ due to person $i$. In this model, we use the following cost function:
\begin{align}
    c_i(z) = g_i(z_i) + \sum_{j \neq i} f_{ij}(z_i, z_j),
\end{align}
where for each person $i$ we have a dimension $m_i \in \Z$ such that $z_i \in \R^{m_i}$. We then have that $f_{ij} : \R^{m_i} \times \R^{m_j} \to \R$ and $g : \R^{m_i} \to \R$. In order to obtain symmetry, we require that $f_{ij}(x, y) = f_{ji}(y, x)$.

It is easy to see by the simple generality of its definition that the \arbitrarysymmetricmodel model is the most general model where the costs between two persons $i$ and $j$ are the same. As an example, to see that it generalizes the \vectormunagalamodel model, we can do the following: marking the cost functions and internal opinions of an instance of \vectormunagalamodel with the superscript "orig," we can take all $m_i$ equal to a fixed $m$, $f_{ij}(x, y) = f_{ij}^{\text{orig}}(x - y)$, and $g_i(x) = g_i^{\text{orig}}(x - s_i^{\text{orig}})$.

In fact, the \arbitrarysymmetricmodel model even generalizes our \ourmodel model. However, surprisingly, our \ourmodel model actually generalizes the \arbitrarysymmetricmodel as well. The proof of this is available in section \ref{section:equivalence_arb_het}.

We therefore have that the \ourmodel model and the \arbitrarysymmetricmodel model generalize each other and so are equivalent. This means that the \ourmodel model is in fact the most general possible model of a symmetric opinion formation game, and so the results we attain for it are the most general of their form. By extension, we have that the \ourmodel model also generalizes  the \fjmodel model, the \vectorfjmodel model, the \quadraticvectormodel model, the \munagalamodel model, and the \vectormunagalamodel model.

\textbf{The \cliquevectormodelcapital  Model.}
We introduce a new model whose costs are identical to the \vectormunagalamodel model, but we consider different dynamics.

In this model, we partition the set of $n$ persons into $k$ \emph{cliques} $C_1, \ldots, C_k$. We assume that the persons inside a given clique $C_i$ will collaborate to minimize the cost to the clique. Specifically, we define a cost function $q_i$ for each clique as follows:

\begin{align*}
    q_i(z) = \sum_{j \in C_i} c_j(z),
\end{align*}

and then define the Nash equilibrium as the state where each clique has minimized its own cost $q_i$; thus, we have $\nabla_{C_i} q_i(x) = 0$ in the Nash equilibrium $x$.

The motivation for this model comes from considering real-life cliques, such as friend groups or groups on social media, where members change their opinions to support the collective good of their group. Notably, while this model is not a direct special case of the \ourmodel model, we demonstrate a reduction from this model to the \ourmodel model in \Cref{section:clique-model}, meaning that all of our results also apply to this model. This reduction involves representing each clique as a single person in the instance of the \ourmodel model, demonstrating the value added by considering heterogeneous opinions. In fact, we can even reduce a clique version of the \ourmodel model to the \ourmodel model itself, meaning that, as the \ourmodel model is the most general possible model of symmetric cost functions for opinion formation, our results are also the most general possible for symmetric cost functions with clique dynamics.

\subsection{Paper Organization}
In \Cref*{section:lb}, we prove that the worst-case price of anarchy across all networks is at least $\frac 2 {e \ln 2}$ for any reasonable cost function, a novel result both in the scalar model and in the multidimensional models we introduce here. The following sections focus more on our multidimensional models. \Cref*{section:general-ub} provides a general upper bound on the price of anarchy of the \ourmodel model for various cost functions, with the bound depending on the $(\lambda, \kappa, p)$-suitability of the cost functions used. In \Cref*{section:tools}, we provide useful tools for finding the closed-form bounds of specific functions. We then apply these tools to find the exact worst-case PoA for functions of the form $\norm{z}^\alpha$ where $\norm{\cdot}$ is any vector-induced norm. Finally, in \Cref*{section:clique-model}, we prove that the \cliquevectormodel model can be reduced to the \ourmodel model.

\section{A Constant Lower Bound on the Worst-Case Price of Anarchy}
\label{section:lb}
In this section, we show \Cref{theorem:lowerbound} restated below, which states that the worst-case PoA of \emph{any} function satisfying the usual requirements is in fact lower bounded by a constant greater than $1$. This result represents a novel lower bound that was previously unknown in the scalar case (i.e. the \munagalamodel model of \tempdel{Bhawalkar et al.\ }\cite{bhawalkar2013coevolutionary}).

\begin{theoremrestate}{\ref{theorem:lowerbound}}
    For any nonzero function $h : \R^m \to \R$ which is nonnegative, differentiable, and convex, the worst-case PoA of $h$ is at least $\frac 2 {e \ln 2}$.
\end{theoremrestate}

As stated in the introduction, this is somewhat surprising: \tempdel{Bhawalkar et al.}\cite{bhawalkar2013coevolutionary}'s result (\Cref{bhawalkarthm})  and our extension of it to the more general model shows that the worst-case PoA of $\norm{z}^\alpha$ (i.e. $\zeta(\alpha)$ in \Cref{eq:zetaintro}) decreases as $\alpha$ goes to $\infty$, which would suggest that by making $\alpha$ sufficiently large we could force the worst-case PoA to approach $1$. Indeed, the fact that this is not the case does not seem to have been known. The value of this limit can in fact be calculated straightforwardly which we do in the below theorem; this value turns out to be exactly $\frac 2 {e \ln 2}$ as illustrated in Figure \ref{fig:boundalpha}.
\begin{theorem}\label{zetalimitthm}
    Define $\zeta(\alpha)$ as in \Cref{eq:zetaintro} to be the worst-case PoA of $\norm{z}^\alpha$ for $\alpha > 1$. Then we have that:
    \begin{align*}
        \lim_{\alpha \to \infty} \zeta(\alpha) = \frac 2 {e\ln 2}.
    \end{align*}
\end{theorem}
\begin{proof}
    Recalling that $\zeta(\alpha) = \frac {(\alpha - 1)^{\alpha - 1}} {\alpha^\alpha} \cdot \frac {\left(2^{\frac \alpha {\alpha - 1}} - 1\right)^\alpha} {2^{\frac \alpha {\alpha - 1}} - 2}$, define $x = \frac 1 {\alpha - 1}$. We will rewrite $\zeta(\alpha)$ in terms of $x$. First, we have that $\alpha - 1 = \frac 1 x$ and $\alpha = \frac 1 x + 1$, so the first term $\frac {(\alpha - 1)^{\alpha - 1}} {\alpha^\alpha}$ becomes $\frac {\left(\frac 1 x\right)^{\frac 1 x}} {\left(\frac 1 x + 1\right)^{\frac 1 x + 1}}$. Second, we again have that $\alpha = \frac 1 x + 1$ and also have that $\frac \alpha {\alpha - 1} = 1 + x$, so the second term $\frac {\left(2^{\frac \alpha {\alpha - 1}} - 1\right)^\alpha} {2^{\frac \alpha {\alpha - 1}} - 2}$ becomes $\frac {\left(2^{1 + x} - 1\right)^{\frac 1 x + 1}} {2^{1 + x} - 2}$. We can then simplify the first term by multiplying by $\frac {x^{\frac 1 x + 1}} {x^{\frac 1 x + 1}}$; this yields:
    \begin{align*}
        \frac {\left(\frac 1 x\right)^{\frac 1 x}} {\left(\frac 1 x + 1\right)^{\frac 1 x + 1}} = \frac {\left(\frac 1 x\right)^{\frac 1 x}} {\left(\frac 1 x + 1\right)^{\frac 1 x + 1}} \cdot \frac {x^{\frac 1 x + 1}} {x^{\frac 1 x + 1}} = \frac {\left(\frac 1 x\right)^{\frac 1 x} x^{\frac 1 x} x} {\left(x\left(\frac 1 x + 1\right)\right)^{\frac 1 x + 1}} = \frac x {(1 + x)^{\frac 1 x + 1}}.
    \end{align*}
    We therefore have $\zeta(\alpha) = \frac x {(1 + x)^{\frac 1 x + 1}} \cdot \frac {\left(2^{1 + x} - 1\right)^{\frac 1 x + 1}} {2^{1 + x} - 2}$. We finally rearrange by switching the denominators to get $\zeta(\alpha) = \frac x {2^{1 + x} - 2} \cdot \frac {\left(2^{1 + x} - 1\right)^{\frac 1 x + 1}} {(1 + x)^{\frac 1 x + 1}}$. We then proceed by independently computing the limit of each term as $\alpha$ goes to $\infty$; that is, as $x$ goes to $0$.
\newcommand*\LHeq{\ensuremath{\overset{\kern2pt L'H}{=}}}
    The limit of the first term can be evaluated using L'Hôpital's rule:
    \begin{align*}
        \lim_{x \to 0} \frac x {2^{1 + x} - 2} \LHeq \lim_{x \to 0} \frac {\frac d {dx} x} {\frac d {dx} \left(2^{1 + x} - 2\right)} = \lim_{x \to 0} \frac 1 {\ln 2 \cdot 2^{1 + x}} = \frac 1 {\ln 2 \cdot 2}.
    \end{align*}
    The limit of the second term can be evaluated by taking its logarithm, and then evaluating the limit of that expression using L'Hôpital's rule:
    \begin{align*}
        \lim_{x \to 0} \ln \left(\frac {\left(2^{1 + x} - 1\right)^{\frac 1 x + 1}} {(1 + x)^{\frac 1 x + 1}}\right) &= \lim_{x \to 0} \frac {\ln \frac {2^{1 + x} - 1} {1 + x}} {\frac x {1 + x}} &\EqComment{Property of logarithm}\\
        & = \lim_{x \to 0} \frac {\frac d {dx} \ln \frac {2^{1 + x} - 1} {1 + x}} {\frac d {dx} \frac x {1 + x}} &\EqComment{L'Hôpital's rule}\\
     & = \lim_{x \to 0} {\frac {1 + x} {2^{1 + x} - 1} \cdot \left(\frac {\ln 2 \cdot 2^{1 + x}} {1 + x} - \frac {2^{1 + x} - 1} {(1 + x)^2}\right)} &\EqComment{Evaluation of derivatives}\\
     & = \frac 1 {2 - 1} \cdot \left(\frac {\ln 2 \cdot 2} 1 - \frac {2 - 1} 1\right) &\EqComment{Value of limit at $x = 0$}\\
     & = 2\ln2 - 1.
    \end{align*}
    We therefore have that $\lim_{x \to 0} \frac {\left(2^{1 + x} - 1\right)^{\frac 1 x + 1}} {(1 + x)^{\frac 1 x + 1}} = e^{2 \ln 2 - 1} = \frac {2^2} e$. Thus, the limit of $\zeta(\alpha)$ as $\alpha$ goes to $\infty$ is $\frac 1 {\ln 2 \cdot 2} \cdot \frac {2^2} e = \frac 2 {e \ln 2}$.
\end{proof}

Though the limit computed in \Cref{zetalimitthm} turns out to be larger than $1$, one might expect that using cost functions that grow faster than polynomial, we could further decrease the worst-case PoA; however, \Cref{lowerboundthm} tells us that this is in fact not the case. Thus, \Cref{zetalimitthm} in a certain sense provides an upper bound that matches the lower bound provided by \Cref{lowerboundthm}: we can say that polynomial cost functions are sufficient to attain the lowest worst-case PoA possible.

The remainder of this section is split into two subsections. Subsection \ref{subsection:lb-exponential} proves \Cref{lowerboundthm} for the specific function $e^x$ -- this proof lays the framework for the full proof of \Cref{lowerboundthm}, as it contains the key ideas motivating \Cref{lowerboundthm} and also introduces two lemmas necessary for the general proof. Subsection \ref{subsection:lb-general} then completes the proof of \Cref{lowerboundthm} for general functions.

\subsection{The Lower Bound for the Exponential Function}
\label{subsection:lb-exponential}
In this subsection, we prove \Cref{lowerboundthm} for the specific function $e^x$ -- this proof will suggest how \Cref{lowerboundthm} can be proven for general functions.
\begin{theorem}\label{expthm}
    The worst-case PoA of $e^x$ is $\frac 2 {e\ln 2}$.
\end{theorem}
The proof of \Cref{expthm} proceeds by showing that any $\lambda, \kappa$ for which $e^x$ is $(\lambda, \kappa)$-suitable must satisfy the below Inequality \ref{inequality:exp-c-lambda-kappa-suit}. From that statement, we then show that $\frac \lambda \kappa \geq \frac 2 {e\ln 2}$, from which we can then apply another lemma to see that the worst-case PoA of $e^x$ is $\frac 2 {e\ln 2}$. The proof of \Cref{lowerboundthm} follows the same structure; as such, the second step is presented in the below \Cref{clemma}, which is then referenced in both proofs:
\begin{lemma}\label{clemma}
The minimum value of $\frac \lambda \kappa$ for $\lambda, \kappa \ge 0$ satisfying (for any $c \in \mathbb{R}$ and $p = 1,2$):
        \begin{align}
            \label{inequality:exp-c-lambda-kappa-suit}
            \frac c p \leq \lambda e^c - \kappa,
        \end{align}
        is $\frac 2 {e \ln 2}$.
\end{lemma}
\begin{proof}
 We first convert the conditions into ones not involving $c$. Each inequality for $p = 1, 2$ is equivalent to the minimum value of $\lambda e^c - \kappa - \frac c p$ being nonnegative. To determine the minimum value of the said expression, we note that its derivative as a function of $c$ is $\lambda e^c - \frac 1 p$, meaning that the minimum value is attained when $c = -\ln (p\lambda)$. The minimum value is therefore $\lambda e^c - \kappa - \frac c p = \frac 1 p - \kappa + \frac {\ln (\lambda p)} p$.
 
 The minimum value being nonnegative, i.e. the inequality $\frac 1 p - \kappa + \frac {\ln (\lambda p)} p \geq 0$, is equivalent to $\kappa \leq \frac {1 + \ln(\lambda p)} p$. This gives two constraints for $p = 1, 2$. The minimum value of $\frac \lambda \kappa$ occurs either at the optimum of $\frac \lambda \kappa$ subject to one of these constraints, or at the intersection.

    If we optimize for one constraint, we get $\kappa = \frac {1 + \ln(\lambda p)} p$, giving us $\frac \lambda \kappa = \frac \lambda {\frac {1 + \ln(\lambda p)} p} = \frac {\lambda p} {1 + \ln(\lambda p)}$.
    
    The minimum value of $\frac x {1 + \ln x}$ occurs at the maximum value of $\frac {1 + \ln x} x$. Differentiating this expression gives $\frac 1 {x^2} - \frac {1 + \ln x} {x^2}$; the maximum value occurs when said derivative is $0$, which can be solved to get $x = 1$.

    Thus, the minimum value of $\frac {\lambda p} {1 + \ln(\lambda p)}$ occurs when $\lambda p = 1$. When $p = 1$, this gives $\lambda = 1$; the $p = 1$ constraint then gives $\kappa \leq 1 + \ln \lambda = 1$, and the $p = 2$ constraint gives $\kappa \leq \frac {1 + \ln(2\lambda)} 2 = \frac {1 + \ln 2} 2 \approx .8466$; this is smaller than $1$, so the optimum for $p = 1$ is irrelevant.

    When $p = 2$, we get $\lambda = \frac 1 2$; the $p = 2$ constraint then gives $\kappa \leq \frac {1 + \ln(2\lambda)} 2 = \frac 1 2$, and the $p = 1$ constraint gives $\kappa \leq 1 + \ln \lambda = 1 - \ln 2 \approx .3066$; this is smaller than $\frac 1 2$, so the optimum for $p = 2$ is irrelevant.

    Therefore, the optimum occurs where the two constraints intersect. This gives the equation $1 + \ln \lambda = \frac {1 + \ln (2\lambda)} 2$, which can be solved by algebra to yield $\ln \lambda = \ln 2 - 1$. This then gives $\lambda = \frac 2 e$, and $\kappa$ can be calculated from the $p = 1$ constraint as $1 + \ln \lambda = \ln 2$. We thus see that the optimal value of $\frac \lambda \kappa $ under the given constraints is $\frac 2 {e\ln 2}$.
\end{proof}
The following lemma will allow us to use the minimum value of $\frac \lambda \kappa$ for which an $\R \to \R$ function is $(\lambda, \kappa)$-suitable to derive a lower bound for its worst-case PoA. This is a modification of a result of \tempdel{Bhawalkar et al.\ }\cite{bhawalkar2013coevolutionary} to show that their proof can be used to get the same result without constraining the cost function to be symmetric. We provide a proof of this lemma in the \Cref{section:poa-lowerbound-inf}; this is mostly a restatement of the argument of \tempdel{Bhawalkar et al.}\cite{bhawalkar2013coevolutionary}, but in one place includes a simpler and more elegant approach that simultaneously removes the reliance on symmetricity. This modification is crucial in order for our results to apply in the full generality of our model, which usefully does not restrict the cost functions to be symmetric. Note that one might expect that the lower bound only happens on large networks, but as shown by \tempdel{Bhawalkar et al.}\cite{bhawalkar2013coevolutionary} and similarly in \Cref{section:poa-lowerbound-inf}, the lower bound can be achieved by a specific network consisting of 3 persons.
\begin{lemma}\label{scalarlowerboundlemma}
    Given a nonnegative, differentiable, and convex function $h : \R \to \R$, the worst-case PoA of $h$ is lower bounded by the infimum of $\frac \lambda \kappa$ for which $\lambda, \kappa > 0$ and $h$ is $(\lambda, \kappa)$-suitable.
\end{lemma}
We now prove \Cref{expthm}.
\begin{proof}[Proof of \Cref{expthm}.]
    We aim to show that $e^x$ being $(\lambda, \kappa, p)$-suitable, i.e. $f'(x)\frac {y - x} p \leq \lambda f(y) - \kappa f(x)$
    for all $x, y \in \R$ where $f(x) = e^x$, is equivalent to $\lambda, \kappa$ satisfying $\frac c p \leq \lambda e^c - \kappa$ for all $c \in \R$. To see this, first, expand the definition $f(x) = e^x$ to see that the suitability inequality is $e^x \frac {y - x} p \leq \lambda e^y - \kappa e^x$. Now divide both sides by $e^x$ to get $\frac {y - x} p \leq \lambda e^{y - x} - \kappa$. We can then let $c = y - x$ to reduce to the inequality $\frac c p \leq \lambda e^c - \kappa$ as desired.

    Therefore, $e^x$ being $(\lambda, \kappa)$-suitable is equivalent to having $\frac c p \leq \lambda e^c - \kappa$ for $c \in \R, p = 1, 2$, meaning that we can apply \Cref{clemma} to see that the minimum value of $\frac \lambda \kappa$ is $\frac 2 {e\ln 2}$. We now first apply \Cref{suitcorol} to get that the worst-case PoA of $e^x$ is at most $\frac 2 {e \ln 2}$. Following that, we apply \Cref{scalarlowerboundlemma} to get that the worst-case PoA of $e^x$ is at least $\frac 2 {e \ln 2}$.
\end{proof}
While the proof of \Cref{clemma} is longer, the key idea comes in the above proof of \Cref{expthm}, where we divide the suitability inequality by $f(x) = e^x$ to get a simpler univariate inequality (Inequality \ref{inequality:exp-c-lambda-kappa-suit}), which can then be analyzed using calculus as in \Cref{clemma}. The same idea is used to prove \Cref{lowerboundthm}, albeit with significantly more technical challenges.

\subsection{The Lower Bound for General Functions}
\label{subsection:lb-general}

We now prove \Cref{creqlemma}, which will show that Inequality \ref{inequality:exp-c-lambda-kappa-suit} is required for all functions. However, its proof relies on \Cref{hreqlemma}, where most of the technical challenge arises; we introduce \Cref{creqlemma} and its proof first in order to explain the necessity of \Cref{hreqlemma}.
\begin{lemma}\label{creqlemma}
    For any function $f : \R \to \R$ which is differentiable, convex, and eventually strictly increasing on $\R^+$, if we have that $f$ is $(\lambda, \kappa, p)$-suitable for some $\lambda, \kappa \in \R, p > 0$, then we must have that $\lambda, \kappa$ satisfy:
    \begin{align*}
        \frac c p \leq \lambda e^c - \kappa,
    \end{align*}
    i.e. Inequality \ref{inequality:exp-c-lambda-kappa-suit}, for all $c \in \R$.
\end{lemma}
\begin{proof}
    As $f$ is both eventually increasing and convex, we have that $\lim_{x \to \infty} f(x) = \infty$. Therefore, there exists some $L$ such that for $x \geq L$, $f(x)$ is both increasing and positive. On this domain, we can write $f(x) = e^{g(x)}$; we now manipulate the suitability equation for $f$ in a similar manner to the proof of \Cref{expthm}:
    \begin{align*}
        f'(x) \frac {y - x} p &\leq \lambda f(y) - \kappa f(x), &\EqComment{$f$ is $(\lambda, \kappa, p)$-suitable}\\
        e^{g(x)}g'(x)\frac {y - x} p &\leq \lambda e^{g(y)} - \kappa e^{g(x)}, &\EqComment{Definition of $g$, chain rule for $e^{g(x)}$}\\
        g'(x)\frac {y - x} p &\leq \lambda e^{g(y) - g(x)} - \kappa. &\EqComment{Dividing by $e^{g(x)}$}
    \end{align*}
    The idea from here is to let $c = g(y) - g(x)$. To facilitate this, note that because $f$ is differentiable and increasing on $[L, \infty)$, and tends to $\infty$, $g$ does as well; this means that $g$ has an inverse $j$ which is also differentiable, increasing, and tends to $\infty$; define $M = g(L)$ so that $j : [M, \infty) \to [L, \infty)$. For some $z \in \R$ to be chosen later set $x = j(z)$ and $y = j(z + c)$. We then have that $g'(x) = \frac 1 {j'(g(x))} = \frac 1 {j'(z)}$, as well as $g(y) - g(x) = (z + c) - z = c$. The suitability equation then reduces to the following:
    \begin{align*}
        \frac 1 {j'(z)} \frac {j(z + c) - j(z)} p \leq \lambda e^c - \kappa.
    \end{align*}
    We now set $h = j'$, which allows us to rewrite the above as:
    \begin{align*}
        \frac 1 p \frac {\int_z^{z + c} h(x)dx} {h(z)} \leq \lambda e^c - \kappa.
    \end{align*}
    As $j$ is differentiable, increasing, and tends to $\infty$, $h$ is continuous and positive on $[M, \infty)$ and satisfies $\int_M^{\infty} h(x)dx = \infty$. $h$ therefore meets the requirements of \Cref{hreqlemma}; it follows that there exist $z$ such that $\frac {\int_z^{z + c} h(x)dx} {h(z)}$ is greater than or arbitrarily close to $c$. This implies that in order to satisfy the above inequality it is necessary to satisfy $\frac c p \leq \lambda e^c - \kappa$ as desired.
\end{proof}
The motivation for the following lemma is that where for the function $e^x$ we can directly reduce the suitability equation to $c \leq \lambda e^c - \kappa$, for a more general function $f$ we have the more complicated term $\frac {\int_z^{z + c} h(x)dx} {h(z)}$. If $h$ were constant, this term would be equal to $c$. Therefore, we intuitively expect that given that $f$ goes to $\infty$, $h$ must either grow so that $\frac {\int_z^{z + c} h(x)dx} {h(z)}$ is actually greater than $c$, or decline to $0$ slowly enough that it is close enough to a constant to make $\frac {\int_z^{z + c} h(x)dx} {h(z)}$ essentially $c$. Specifically, if $h$ declined geometrically in order to cause $\frac {\int_z^{z + c} h(x)dx} {h(z)}$ to be less than $rc$ for some $r < 1$, then its integral would be finite, in contradiction to the constraints of the lemma.
\begin{lemma}\label{hreqlemma}
    For any function $h$ : $[M, \infty)$ (for some $M \in \R$) which is continuous, positive, and satisfies $\int_M^{\infty} h(x)dx = \infty$, we have that for any $c \in R$ and $\epsilon > 0$, there exist $z \in [\max(M, M - c), \infty)$ such that:
    \begin{align*}
        \frac {\int_z^{z + c} h(x)dx} {h(z)} \geq c - \epsilon.
    \end{align*}
\end{lemma}
\begin{proof}
    First note that the case of $c = 0$ is trivial, as we have that $\frac {\int_z^{z + c} h(x)dx} {h(z)} = 0 = c$. We then note that it suffices to prove the desired statement for $c = \pm 1$. To see why, assume that we have proven the statement for $c = \pm 1$ for all relevant functions $h$. Then for any other $c \neq 0$, we can consider the dilation $h_2(x) = h(|c|x)$, and observe that the equation $\frac {\int_z^{z + c} h(x)dx} {h(z)} \geq c - \epsilon$ is equivalent to $\frac {\int_w^{w + 1} h_2(x)dx} {h_2(w)} \geq \frac c {|c|} - \frac \epsilon {|c|}$ where $w = \frac z {|c|}$. Then, because $\frac c {|c|} = \pm 1$, we know that for all $\frac \epsilon {|c|} > 0$ there exists some $w$ satisfying the inequality, we can conclude that for all $\epsilon > 0$, we can take $z = |c|w$.

    We now split into cases for each of $c = -1, 1$; these are handled by the below lemmas respectively.
    \begin{lemma}\label{cneglemma}
    For any function $h$ : $[M, \infty)$ (for some $M \in \R$) which is continuous, positive, and satisfies $\int_M^{\infty} h(x)dx = \infty$, we have that for any $\epsilon > 0$, there exist $z \in [M + 1, \infty)$ such that
    \begin{align*}
        \frac {\int_z^{z - 1} h(x)dx} {h(z)} \geq -1 - \epsilon.
    \end{align*}
    \end{lemma}
    \begin{lemma}\label{cposlemma}
    For any function $h$ : $[M, \infty)$ (for some $M \in \R$) which is continuous, positive, and satisfies $\int_M^{\infty} h(x)dx = \infty$, we have that for any $\epsilon > 0$, there exist $z \in [M, \infty)$ such that:
    \begin{align*}
        \frac {\int_z^{z + 1} h(x)dx} {h(z)} \geq 1 - \epsilon.
    \end{align*}
    \end{lemma}
    The full proof of the above lemmas can be found in \Cref{appendix:cposneglemma-proof}. With these two lemmas the proof in complete.
\end{proof}

The following last lemma allows us to extend lower bounds given by \Cref{scalarlowerboundlemma} for scalar functions to apply to multidimensional functions as well.
\begin{lemma}\label{restrictionlemma}
    Given a function $h(z) : \R^m \to \R$, consider its restriction along a particular direction; that is, for some unit vector $v \in \R^m$, consider $u(x) = h(xv)$, so that $u : \R \to \R$. The worst-case PoA of $h$ is at least the worst-case PoA of $u$.
\end{lemma}
\begin{proof}
Any opinion formation game considered in the definition of the worst-case PoA of $u$ has costs defined as follows:
\begin{align*}
    c_i(z) = r_i u(R_i z_i - s_i) + \sum_{j \neq i} w_{ij}u(A_{ij} z_i + B_{ij}z_j).
\end{align*}
$c_i(z)$ can then alternatively be written using $h$:
\begin{align*}
    c_i(z) = r_i h(R_i v z_i - s_i v) + \sum_{j \neq i} h_{ij}u(A_{ij} v z_i + B_{ij} v z_j).
\end{align*}
Any such game is therefore also considered in the definition of the worst-case PoA of $h$. Therefore, the supremum of PoAs of such games (which is the definition of the worst-case PoA of $u$) must also be a lower bound on the PoA of $h$.
\end{proof}

We may now finally prove \Cref{lowerboundthm}. For clarity, we first restate it below:
\begin{theoremrestate}{\ref{lowerboundthm}}
    For any nonconstant function $h : \R^m \to \R$ which is nonnegative, differentiable, and convex, the worst-case PoA of $h$ is at least $\frac 2 {e \ln 2}$.
\end{theoremrestate}
\begin{proof}
    We first prove the statement for functions $f : \R \to \R$. First note that because $f$ is convex and nonzero, it must be that either $f$ is eventually strictly increasing as $x$ goes to $+\infty$, or eventually strictly decreasing as $x$ goes to $-\infty$ (usually, both are true). If only the latter is the case, then we can instead consider the function $x \mapsto f(-x)$, because the definition of worst-case PoA is maintained under dilation of the function as the dilation can be canceled out by negating $A_{ij}, B_{ij}, R_{ij}$. Therefore, we will assume that $f$ is eventually strictly increasing.
    
    Thus, $f$ is eventually increasing. Consider the $(\lambda, \kappa)$ for which $f$ is $(\lambda, \kappa)$-suitable. By \Cref{creqlemma}, because $f$ is $(\lambda, \kappa, p)$-suitable for $p = 1, 2$, it must satisfy $\frac c p \leq \lambda e^c - \kappa$ for all $c \in \R$ for $p = 1, 2$. Then, by \Cref{clemma}, $\frac \lambda \kappa$ must be at least $\frac 2 {e \ln 2}$. Thus, for any $\lambda, \kappa > 0$ such that $f$ is $(\lambda, \kappa)$-suitable, we have that $\frac \lambda \kappa \geq \frac 2 {e \ln 2}$. We can then apply \Cref{scalarlowerboundlemma} to get that the worst-case PoA of $f$ is lower bounded by the infimum of such $\frac \lambda \kappa$. It follows that the worst-case PoA of $f$ is at least $\frac 2 {e \ln 2}$.

    To extend this result to $h : \R^m \to \R$, we first note that if $h$ is nonconstant, then there exists some restriction of $h$ along a particular direction, i.e. some $u(x) = h(xv)$ for a unit vector $v \in \R^m$, such that $u$ is nonconstant. We then have that the worst-case PoA of $u$ is at least $\frac 2 {e \ln 2}$, so by \Cref{restrictionlemma} the worst-case PoA of $h$ is at least the worst-case PoA of $u$ and therefore at least $\frac 2 {e \ln 2}$.
\end{proof}

In addition to the surprising implications of this theorem touched upon at the beginning of this section, it also allows us to give the exact worst-case PoA for a large class of functions derived from $e^x$. As a particularly interesting example, we are able to determine that the worst-case PoA of $\cosh x$ is exactly $\frac 2 {e\ln 2}$. These corollaries are fully described later, in \Cref{subsection:exact-poa}, after we have developed the tools to prove matching upper bounds on the worst-case PoA for those functions.

\section{General Upper Bound on the Price of Anarchy}
\label{section:general-ub}

In this section, we will mainly prove \Cref{suitabilitythm}, restated below.

\begin{theoremrestate}{\ref{suitabilitythm}}
    Consider a symmetric opinion formation game whose cost functions $f_{ij}, g_i$ are nonnegative and differentiable. If for some $\lambda, \kappa > 0$, we have that all $f_{ij}$ are $(\lambda, \kappa, 2)$-suitable and all $g_i$ are $(\lambda, \kappa, 1)$-suitable, then the price of anarchy of the opinion formation game is at most $\frac \lambda \kappa$.
\end{theoremrestate}

The above theorem is proven for the scalar case (i.e. when $m = 1$) by \tempdel{Bhawalkar et al.\ }\cite{bhawalkar2013coevolutionary}. We will now use the remainder of this section to extend their result to our more general model.
\begin{proof}[Proof of \Cref{suitabilitythm}]
    The proof is structured as follows: we first prove the following inequality (recall that $c_i(z)$ defines the penalty of a set of expressed opinions $z$ for person $i$, and $SC(z) = \sum_i c_i(z)$ defines the social cost of $z$):
    \begin{align}
        \label{equation:roughgarden-inequality}
        \sum_i \nabla_i c_i(x)^T (y_i - x_i) \leq \lambda SC(y) - \kappa SC(x).
    \end{align}
    We then use a result of Roughgarden and Schoppmann \cite{roughgarden2015local} to derive a bound on the PoA from inequality \ref{equation:roughgarden-inequality}. The idea behind the proof of Inequality \ref{equation:roughgarden-inequality} is to divide it into smaller inequalities for each $i$ (corresponding to the cost due to the difference between a person's internal and expressed opinion) and each pair $(i, j)$ (corresponding to the cost due to the difference between two persons' expressed opinions), simplify the inequality corresponding to pairs by making use of the fact that the game is symmetric, then apply the fact that the cost functions are suitable to prove the resulting inequalities.

    Correspondingly, we define the following auxiliary variables:
    \begin{align*}t_{ij} & = (\nabla_i f_{ij}(A_{ij}x_i + B_{ij}x_j))^T(y_i - x_i) + (\nabla_j f_{ji}(A_{ji}x_j + B_{ji}x_i))^T(y_j - x_j),                         \\
        u_{ij} & = \lambda f_{ij}(A_{ij}y_i + B_{ij}y_j) + \lambda f_{ji}(A_{ji}y_j + B_{ji}y_i) - \kappa f_{ij}(A_{ij}x_i + B_{ij}x_j) - \kappa f_{ji}(A_{ji}x_j + B_{ji}x_i), \\
        v_i    & = (\nabla_i g_i(R_i x_i - s_i))^T(y_i - x_i),                                                                       \\
        w_i    & = \lambda g_i(R_i y_i - s_i) - \kappa g_i(R_i x_i - s_i).
    \end{align*}
    
    The following two lemmas show how we can express the components of the local smoothness inequality in terms of the above auxiliary variables.
    \begin{lemma}\label{suithelper1}
        \begin{align*}
            \sum_i \nabla_i c_i(x)^T(y_i - x_i)
             & = \sum_i v_i + \sum_{i < j} t_{ij}.
        \end{align*}
    \end{lemma}
    \begin{proof}
        First recall that the cost function for person $i$ is $c_i(z) = g_i(R_i z_i - s_i) + \sum_{j \neq i} f_{ij}(A_{ij} z_i + B_{ij} z_j)$. We can thus write:
        \begin{align*}
            \sum_i \nabla_i c_i(x)^T(y_i - x_i) = \sum_i\left[(\nabla_i g_i(R_i x_i - s_i))^T(y_i - x_i) + \sum_{j \neq i} (\nabla_i f_{ij}(A_{ij} x_i + B_{ij} x_j))^T(y_i - x_i)\right],
        \end{align*}
        by the definition of $c_i(x)$. Then, extracting the inner sum we have:
        \begin{align*}
            \sum_i \nabla_i c_i(x)^T(y_i - x_i) = \sum_i (\nabla_i g_i(R_i x_i - s_i))^T(y_i - x_i) + \sum_{i, j}(\nabla_i f_{ij}(A_{ij} x_i + B_{ij} x_j))^T(y_i - x_i).
        \end{align*}
        We then have that $\sum_i (\nabla_i g_i(R_i x_i - s_i))^T(y_i - x_i) = \sum_i v_i$ by definition. The other term can be manipulated by combining the term for $(i, j)$ with that for $(j, i)$:
        \begin{align*}
            & \sum_{i, j}(\nabla_i f_{ij}(A_{ij} x_i + B_{ij} x_j))^T(y_i - x_i) \\
             & = \sum_{i < j}(\nabla_i f_{ij}(A_{ij} x_i + B_{ij}x_j))^T(y_i - x_i) + (\nabla_j f_{ji}(A_{ji} x_j + B_{ji} x_i))^T(y_j - x_j) \\
             & = \sum_{i < j} t_{ij}.
        \end{align*}
        We therefore have that $\sum_{i, j}(\nabla_i f_{ij}(x_i - x_j))^T(y_i - x_i) = \sum_{i < j} t_{ij}$ and so $\sum_i \nabla_i c_i(x)^T(y_i - x_i) = \sum_i v_i + \sum_{i < j} t_{ij}$ as desired.
    \end{proof}

    \begin{lemma}\label{suithelper2}
        \begin{align*}
            \lambda SC(y) - \kappa SC(x)
             & = \sum_i w_i + \sum_{i < j} u_{ij}.
        \end{align*}
    \end{lemma}
    \begin{proof}
        Recall that $SC(z) = \sum_i c_i(z)$; we first expand $\lambda SC(y) - \kappa SC(x)$ using the definition of $c_i(z)$:
        \begin{align*}
            \lambda SC(y) - \kappa SC(x) = \lambda \sum_i c_i(y) - \kappa \sum_i c_i(x).
        \end{align*}
        Then, expanding the definition of $c_i(x)$ we have:
        \begin{align*}
            \lambda SC(y) - \kappa SC(x) = &\lambda \sum_i \left[g_i(R_i y_i - s_i) + \sum_{j \neq i} f_{ij}(A_{ij} y_i + B_{ij} y_j)\right] \\ &- \kappa \sum_i \left[g_i(R_i x_i - s_i) + \sum_{j \neq i} f_{ij}(A_{ij} x_i + B_{ij} x_j)\right].
        \end{align*}
        By extracting inner sums and grouping together like sums, we get:
        \begin{align*}
            \lambda SC(y) - \kappa SC(x) = \sum_i \left[\lambda g_i(R_i y_i - s_i) - \kappa g_i(R_i x_i - s_i)\right] + \sum_{i, j}\left[\lambda f_{ij}(A_{ij} y_i + B_{ij} y_j) - \kappa f_{ij}(A_{ij} x_i + B_{ij} x_j)\right].
        \end{align*}
        We then have that $\sum_i \left[\lambda g_i(R_i y_i - s_i) - \kappa g_i(R_i x_i - s_i)\right] = \sum_i w_i$ by definition. The other term can be manipulated by combining the term for $(i, j)$ with that for $(j, i)$:
        \begin{align*}
            & \sum_{i, j}\left[\lambda f_{ij}(A_{ij} y_i + B_{ij} y_j) - \kappa f_{ij}(A_{ij} x_i + B_{ij} x_j)\right]\\
             & = \sum_{i < j}\left[\lambda f_{ij}(A_{ij} y_i + B_{ij} y_j) + \lambda f_{ji}(A_{ij} y_j + B_{ij} y_i) - \kappa f_{ij}(A_{ij} x_i + B_{ij} x_j) - \kappa f_{ji}(A_{ij}x_j + B_{ij} x_i)\right] \\
             & = \sum_{i < j} u_{ij}.
        \end{align*}
        We therefore have that $\sum_{i, j}\left[\lambda f_{ij}(A_{ij} y_i + B_{ij} y_j) - \kappa f_{ij}(A_{ij} x_i + B_{ij} x_j)\right] = \sum_{i < j} u_{ij}$, and so $\lambda SC(y) - \kappa SC(x)
            = \sum_i w_i + \sum_{i < j} u_{ij}$ as desired.
    \end{proof}
    The following two lemmas prove the key inequalities necessary to show Inequality \ref{equation:roughgarden-inequality}. \Cref{suithelper3} essentially follows directly from $(\lambda, \kappa, 1)$-suitability.
    \begin{lemma}\label{suithelper3}
        For all persons $i$, we have $v_i \leq w_i$.
    \end{lemma}
    \begin{proof}
        We apply the suitability inequality for $g_i$ as follows:
        \begin{align*}
            v_i
            &= (\nabla_i g_i(R_i x_i - s_i))^T R_i(y_i - x_i) &\EqComment{Definition of $v_i$}\\
            &= (R_i^T \nabla g_i(R_i x_i - s_i))^T (y_i - x_i)&\EqComment{Chain rule for $\nabla_i$}\\
            &= \nabla g_i(R_i x_i - s_i)^T R_i (y_i - x_i)&\EqComment{Expanding transpose}\\
            &= \nabla g_i(R_i x_i - s_i)^T ((R_i y_i - s_i) - (R_i x_i - s_I)) &\EqComment{Adding and subtracting $s_i$}\\
            &\leq \lambda g_i(R_i y_i - s_i) - \kappa g_i(R_i x_i - s_i)&\EqComment{$g_i$ is $(\lambda, \kappa, 1)$-suitable}\\
            &= w_i. &\EqComment{Definition of $w_i$}
        \end{align*}
    \end{proof}
    \begin{lemma}\label{suithelper4}
        For all persons $i, j$, we have $t_{ij} \leq u_{ij}$.
    \end{lemma}
    \begin{proof}
        We first manipulate $t_{ij}$ as follows, noting that the symmetry of opinion formation game implies that $A_{ji} = B_{ij}$ and $B_{ji} = A_{ij}$:
        \begin{align*}
            t_{ij}
             & = (\nabla_i f_{ij}(A_{ij} x_i + B_{ij} x_j))^T(y_i - x_i) + (\nabla_j f_{ji}(A_{ji} x_j + B_{ji} x_i))^T(y_j - x_j) & \EqComment{Definition of $t_{ij}$} \\
             & = (\nabla_i f_{ij}(A_{ij} x_i + B_{ij} x_j))^T(y_i - x_i) + (\nabla_j f_{ij}(B_{ij} x_j + A_{ij} x_i))^T(y_j - x_j) &\EqComment{Symmetry}\\
             & = (A_{ij}^T\nabla f_{ij}(A_{ij} x_i + B_{ij} x_j))^T(y_i - x_i) + (B_{ij}^T \nabla f_{ij}(B_{ij} x_j + A_{ij} x_i))^T(y_j - x_j) &\EqComment{Chain rule}  \\
             & = \nabla f_{ij}(A_{ij} x_i + B_{ij} x_j)^T A_{ij}(y_i - x_i) + \nabla f_{ij}(B_{ij} x_j + A_{ij} x_i)^T B_{ij}(y_j - x_j) &\EqComment{Transpose}  \\
             & = \nabla f_{ij}(A_{ij} x_i + B_{ij} x_j)^T((A_{ij} y_i + B_{ij}y_j) - (A_{ij}x_i + B_{ij}x_j)).
        \end{align*}
        We then observe that in the definition of $u_{ij}$, by symmetry we have that the terms $\lambda f_{ij}(A_{ij} y_i + B_{ij} y_j)$ and $\lambda f_{ji}(A_{ji} y_j + B_{ji} y_i)$ are equal, and the other two terms are similarly equal as well, meaning that:
        \begin{align*}
            u_{ij}
             & = \lambda f_{ij}(A_{ij} y_i + B_{ij} y_j) + \lambda f_{ji}(A_{ji} y_j + B_{ji} y_i) - \kappa f_{ij}(A_{ij} x_i + B_{ij} x_j) - \kappa f_{ji}(A_{ji} x_j + B_{ji} x_i) \\
             & = 2\left[\lambda f_{ij}(A_{ij} y_i + B_{ij} y_j) - \kappa f_{ij}(A_{ij} x_i + B_{ij} x_j)\right].
        \end{align*}
        The desired inequality can now be expressed as:
        \begin{align*}
        \nabla f_{ij}(A_{ij} x_i + B_{ij} x_j)^T((A_{ij} y_i + B_{ij}y_j) - (A_{ij}x_i + B_{ij}x_j)) \leq 2\left[\lambda f_{ij}(A_{ij} y_i + B_{ij} y_j) - \kappa f_{ij}(A_{ij} x_i + B_{ij} x_j)\right].
        \end{align*}
        If we let $a = A_{ij} x_i + B_{ij} x_j$ and $b = A_{ij} y_i + B_{ij} y_j$, then this is:
        \begin{align*}
            (\nabla_i f_{ij}(a))^T (b - a)
             & \leq 2\left[\lambda f_{ij}(b) - \kappa f_{ij}(a)\right],
        \end{align*}
        i.e. $(\nabla_i f_{ij}(a))^T \frac {b - a} p
            \leq \lambda f_{ij}(b) - \kappa f_{ij}(a)$, which is precisely the definition of $f_{ij}$ being $(\lambda, \kappa, 2)$-suitable, so we are done.
    \end{proof}
    The key idea in \Cref{suithelper4} is that the symmetry of the opinion formation game causes the inequality to collapse to the definition of $(\lambda, \kappa, 2)$-suitability. 
    
    We can now show Inequality \ref{equation:roughgarden-inequality}:
    \begin{align*}
        \sum_i \nabla_i c_i(x)^T (y_i - x_i)
         & = \sum_i v_i + \sum_{i < j} t_{ij}    & \EqComment{\Cref{suithelper1}}                             \\
         & \leq \sum_i w_i + \sum_{i < j} u_{ij} & \EqComment{Lemmas \ref{suithelper3} and \ref{suithelper4}} \\
         & = \lambda SC(y) - \kappa SC(x).       & \EqComment{\Cref{suithelper2}}
    \end{align*}
    The purpose of showing Inequality \ref{equation:roughgarden-inequality} is to apply the framework of Roughgarden and Schoppmann \cite{roughgarden2015local}. Roughgarden and Schoppmann developed a local smoothness framework for analyzing cost-minimization games, building on the work of \tempdel{Blum et al.\ }\cite{blumregret}. Their result, which depends on Inequality \ref{equation:roughgarden-inequality}, is expressed in the following lemma:
    \begin{lemma}
        \label{localsmoothness}
        Consider a general cost-minimization game where we have $n$ players, each of whom can take some action represented by a vector $z_i \in \R^m$ (defining $z = (z_1, \ldots, z_n)$ as before), for which they have some penalty $c_i(z)$, such that the social cost is $SC(z) = \sum_i c_i(z)$. Assume that $SC(z)$ is always nonnegative.

        If for all $n$-tuples $x, y$ such that $x_i, y_i \in \R^n$, we have:
        \begin{align*}
            \sum_i \nabla_ic_i(x)^T(y_i - x_i) \leq \lambda SC(y) - \kappa SC(x),
        \end{align*}
        then the PoA of the cost minimization game is at most $\frac \lambda \kappa$.
    \end{lemma}
    As our notation is slightly different from that of Roughgarden and Schoppmann, we provide a proof of the above lemma in the \Cref{appendix:localsmoothness-proof} for clarity. We can now complete the proof as follows. First, because all cost functions $f_{ij}, g_i$ are nonnegative, the social cost $SC(z)$ is always nonnegative. Second, we have shown Inequality \ref{equation:roughgarden-inequality}, which is identical to the inequality appearing in Lemma \ref{localsmoothness}. Therefore, the conclusion of Lemma \ref{localsmoothness} applies, and we have that the PoA of the symmetric opinion formation game is at most $\frac \lambda \kappa$.
\end{proof}

\section{The Worst-Case Price of Anarchy for Specific Functions}
\label{section:tools}

This section deals with the worst-case PoA of a function $h$ as defined in the introduction. We divide into two subsections: subsection \ref{subsection:poa-tools} demonstrates a variety of tools for proving upper bounds on the worst-case PoA of a function $h$, while subsection \ref{subsection:exact-poa} shows the exact worst-case PoA for a large class of functions.

\subsection{Tools for Upper Bounding the Worst-Case PoA}
\label{subsection:poa-tools}

We will provide four key theorems relating the $(\lambda, \kappa, p)$-suitability of one function to that of another. The primary objective of these theorems is to extend bounds on the worst-case PoA of a function of a scalar to bounds on the worst-case PoA of a function of a vector, but they have applications even in the case of scalars; this will become relevant in later sections. The first of these theorems is essential but easy to see:

\begin{theorem}\label{suitcombo}
    A nonnegative linear combination of $(\lambda, \kappa, p)$-suitable functions is $(\lambda, \kappa, p)$-suitable. That is, given functions $f, g : \R^m \to \R$, for any $a, b \geq 0$ define $h : \R^m \to \R$ by $h(z) = af(z) + bg(z)$. If $f, g$ are $(\lambda, \kappa, p)$-suitable, then $h$ is as well; as a result, if $f, g$ are $(\lambda, \kappa)$-suitable, then $h$ is as well.
\end{theorem}
\begin{proof}
    The inequality defining $(\lambda, \kappa, p)$-suitability is $\nabla h(x)^T \frac {y - x} p \leq \lambda h(y) - \kappa h(x)$, which can be rephrased as $\lambda h(y) - \kappa h(x) - \left(\nabla h(x)^T \frac {y - x} p\right) \geq 0$. The quantity $\lambda h(y) - \kappa h(x) - \left(\nabla h(x)^T \frac {y - x} p\right)$ is linear in the relevant function, so if its values for $f, g$ are $q, r \geq 0$ respectively, then its value for $h$ will be $aq + br \geq 0$, making $h$ $(\lambda, \kappa, p)$-suitable.
\end{proof}

The next of these theorems is similar, showing that $(\lambda, \kappa, p)$-suitability is maintained under composition with an affine transformation.
\begin{theorem}\label{suitlinear}
    Given a function $f : \R^m \to \R$, consider for any positive integer $l$, $m \times l$ real matrix $A$, and vector $v \in \R^m$, the function $h : \R^l \to \R$ defined by $h(z) = f(Az + v)$. If for some $\lambda, \kappa, p > 0$, $f$ is $(\lambda, \kappa, p)$ suitable, then $h$ is also $(\lambda, \kappa, p)$-suitable. As a result, if $f$ is $(\lambda, \kappa)$-suitable, then $h$ is as well.
\end{theorem}
\begin{proof}
    We desire to show that $\nabla h(x)^T \frac {y - x} p \leq \lambda h(y) - \kappa h(x)$ for all $x, y \in \R^l$; this can be seen in the following steps:
    \begin{align*}
        \nabla h(x)^T \frac {y - x} p
         & = \nabla (f(Ax + v))^T \frac {y - x} p       & \EqComment{Definition of $h$}                         \\
         & = (A^T (\nabla f)(Ax + v))^T \frac {y - x} p & \EqComment{Chain rule}                                \\
         & = (\nabla f(Ax + v))^TA \frac {y - x} p       & \EqComment{Property of transpose}     \\
         & = (\nabla f(Ax + v))^T  \frac {(Ay + v) - (Ax + v)} p     & \EqComment{Rearranging}                               \\
         & \leq \lambda f(Ay + v) - \kappa f(Ax + v)              & \EqComment{$(\lambda, \kappa, p)$-suitability of $f$} \\
         & = \lambda h(y) - \kappa h(x).                   & \EqComment{Definition of $h$}
    \end{align*}
\end{proof}
This theorem, though still relatively simple, is already significant as it allows us to attain results for the \quadraticvectormodel model discussed in \Cref{eq:quadratic-vectorized:cost}. Recall that the model has the following cost function:
\begin{align*}c_i(z) = (z_i - s_i)^T R_i (z_i - s_i) + \sum_j (z_i - z_j)^T W_{ij} (z_i - z_j).\end{align*}
where $R_i$ is positive semidefinite (not required to be positive definite) for all $i$, and $W_{ij} = W_{ji}$ is positive semidefinite for all $i,j$. As discussed earlier, when $A$ is positive semidefinite, $z^T A z$ can be rewritten as $\norm{Sz}^2$, where $S = A^{\frac 1 2}$. Therefore under the \vectormunagalamodel model, the cost functions $f_{ij}, g_i$ are of the form $h(z) = \norm{Sz}^2$. By \Cref{bhawalkarthm}, there exist $\lambda, \kappa > 0$ such that the scalar-to-scalar function $f(x) = x^2$ is $(\lambda, \kappa)$-suitable and $\frac \lambda \kappa = \frac 9 8$. We then have by \Cref{suitlinear} that any function of the form $\norm{Sz}^2$ is also $(\lambda, \kappa)$-suitable for the said $\lambda, \kappa$, and so by \Cref{suitabilitythm}, the PoA for undirected instances of the \quadraticvectormodel model is upper bounded by $\frac \lambda \kappa = \frac 9 8$.

The next theorem shows how $(\lambda, \kappa)$-suitability is maintained under more general transformations. This is the first result in this paper which demands convexity.
\begin{theorem}\label{suitconvex}
    Given a function $f : \R \to \R$ which is nondecreasing on some subset $S$ of $\R$, and a \textit{convex} and differentiable function $g : \R^m \to S$, define $h: \R^m \to \R$ by $h(z) = f(g(z))$. If for some $\lambda, \kappa, p > 0$, $f$ is $(\lambda, \kappa, p)$ suitable, then $h$ is also $(\lambda, \kappa, p)$-suitable. As a result, if $f$ is $(\lambda, \kappa)$-suitable, then $h$ is as well.
\end{theorem}
\begin{proof}
    We first show that as $g$ is convex, for all $x, y \in \R^m$ we have:
    \begin{align}\label{ineq11}\nabla g(x)^T (y - x)\leq g(y) - g(x).\end{align}
    To see this, define $u : [0, 1] \to \R$ by $u(t) = g(x + t(y - x))$. Then $u$ is convex, and $h'(t) = (\nabla g)(x + t(y - x))^T (y - x)$, meaning that we would like to show $h'(0) \leq h(1) - h(0)$. Because $u$ is convex, for $t \geq 0$ we have $h'(t) \geq h'(0)$, so $h(1) - h(0) = \int_0^1 h'(z)dz \geq \int_0^1 h'(0)dz = h'(0)$
    as desired.

    We proceed to show that $h$ is $(\lambda, \kappa, p)$-suitable. We use here the fact that because $f$ is nondecreasing on the codomain of $S$, we have that $f'(g(x))$ is nonnegative.
    \begin{align*}
        \nabla h(x)^T \frac {y - x} p
         & = \nabla (f(g(x))^T \frac {y - x} p        & \EqComment{Definition of $h$}                          \\
         & = f'(g(x)) (\nabla g(x))^T \frac {y - x} p & \EqComment{Chain rule}                                 \\
         & = f'(g(x)) \frac {\nabla g(x)^T (y - x)} p & \EqComment{Rearranging}                                \\
         & \leq f'(g(x)) \frac {g(y) - g(x)} p              & \EqComment{Inequality \ref{ineq11}, nonnegativity of $f'(g(x))$} \\
         & \leq \lambda f(g(y)) - \kappa f(g(x))            & \EqComment{$(\lambda, \kappa, p)$-suitability of $f$}  \\
         & = \lambda h(y) - \kappa h(x).                    & \EqComment{Definition of $h$}
    \end{align*}
\end{proof}
The key idea of the proof is that the convexity of $g$ allows us to prove the below Inequality \ref{ineq11}, bounding an expression containing the gradient of $g$ into a simple difference and so transforming the suitability inequality for $h$ into the suitability inequality for $f$.

The primary value of \Cref{suitconvex} is in motivating the following \Cref{suitnorm}, which is more directly useful for proving PoA upper bounds in our \vectormunagalamodel model. However, \Cref{suitconvex} is useful even in the \munagalamodel model. As an example, we will show in Section \ref{section:lb} that the worst-case PoA of $e^x$ is at most $\frac 2 {e\ln 2}$ as there exist $\lambda, \kappa > 0$ with $\frac \lambda \kappa = \frac 2 {e\ln 2}$ such that $e^x$ is $(\lambda, \kappa)$-suitable. If we then let $f(x) = e^x$ and $g(x) = x^2$, we can apply \Cref{suitconvex} to get that $e^{x^2}$ is also $(\lambda, \kappa)$-suitable, and so its worst-case PoA is also at most $\frac 2 {e\ln 2}$. This is an extremely useful result as said bound turns out to be tight.

The final theorem in this section is very closely related to \Cref{suitconvex}; the key difference is that the function $g$ is not required to be differentiable at the origin, which means that we can use it when $g$ is a \textit{norm}, a use-case which is natural in our multidimensional model.
\begin{theorem}\label{suitnorm}
    We are given a function $f : \R \to \R$ and a function $g : \R^m \to \R_{\ge 0}$ which satisfy the following constraints:
    \begin{itemize}
        \item $g$ is convex, and differentiable everywhere, possibly excluding the origin \footnote{Note that $g$ being differentiable outside of the origin and convex implies that it is continuous everywhere.}
        \item $g(0) = 0$, $f'(0) = 0$, $f'(z) \geq 0$ for $z \geq 0$
    \end{itemize}

    Define $h: \R^m \to \R$ by $h(z) = f(g(z))$. If for some $\lambda, \kappa, p > 0$, $f$ is $(\lambda, \kappa, p)$ suitable, then $h$ is also $(\lambda, \kappa, p)$-suitable. As a result, if $f$ is $(\lambda, \kappa)$-suitable, then $h$ is as well.
\end{theorem}
\begin{proof}
    We need to show the suitability inequality, i.e. $\nabla h(x)^T \frac {y - x} p \leq \lambda h(y) - \kappa h(x)$. When $x \neq 0$, we note that all the constraints required by \Cref{suitconvex} hold ($g$ is convex, and we can take $S = \R_{\ge 0}$), and so we can repeat the same argument. It therefore suffices to prove the suitability inequality for the case when $x = 0$.

    To handle this case, we will show that $\nabla h(0)$ is in fact zero. We first argue that because $g$ is convex, $\norm{\nabla g(x)}$ for $0 < \norm{x} < 1$ must be bounded by the maximum $\norm{\nabla g(x)}$ for $\norm{x} = 1$, which exists as the set $\{x\ |\ \norm{x} = 1\}$ is closed.

    To see that this bound holds, note that for any such $x$, we can let $v$ be a unit vector in the direction of $\nabla g(x)$, and consider the restriction of $g$ along that direction, i.e. the function $u(t) = g(x + tv)$. $u$ must also be convex, meaning that $u'$ must be increasing. Let $t_1 \leq 0$ and $t_2 \geq 0$ satisfy $\norm{x + t_1v} = 1, \norm{x + t_2v} = 1$. Such $t_1, t_2$ must exist because $\norm{x} \leq 1$, while $\norm{x + tv}$ goes to $+\infty$ for arbitrarily positive and arbitrarily negative $t$. We then have that $u'(t_1) \leq u'(0) \leq u'(t_2)$. This implies that $|u'(0)|$ is bounded by the larger of $|u'(t_1)|, |u'(t_2)|$, which is bounded by the larger of $\norm{\nabla g(x + t_1v)}, \norm{\nabla g(x + t_2v)}$. Therefore, because $u$ is defined in the direction of $\nabla g(x)$, we have that $\norm{\nabla g(x)} = |u'(0)|$ and so is bounded by $\norm{\nabla g(x')}$ for some $\norm{x'} = 1$ as desired.

    We therefore have that for $x$ close to $0$, $\norm{\nabla g(x)}$ is bounded by some constant $C$. We can then argue using a limit:
    \begin{align*}
        \lim_{x \to 0} \norm{\nabla h(x)}
         & = \lim_{x \to 0} \norm{f'(g(x))\nabla g(x)} & \EqComment{Chain rule}           \\
         & \leq C\lim_{x \to 0} |f'(g(x))|        & \EqComment{Boundedness argument} \\
         & = C |f'(0)|                            & \EqComment{$g(0) = 0$}           \\
         & = 0.                                   & \EqComment{$f'(0) = 0$}
    \end{align*}
    We can now show that $(\lambda, \kappa, p)$-suitability holds at $x = 0$:
    \begin{align*}
        \nabla h(x)^T \frac {y - x} p
         & = 0                                   & \EqComment{As just shown}                             \\
         & = f'(g(x))\frac {g(y) - g(x)} p       & \EqComment{$g(0) = 0$ and $f'(0) = 0$}                \\
         & \leq \lambda f(g(y)) - \kappa f(g(x)). & \EqComment{$(\lambda, \kappa, p)$-suitability of $f$}
    \end{align*}
\end{proof}
As an important special case of this theorem, if $g$ is a norm and is differentiable everywhere other than the origin, then it meets the above requirements. The idea for the proof for the above theorem is to repeat the proof of \Cref{suitconvex}, using the fact that $g(0) = 0$ and $f'(0) = 0$ to account for the fact that $g$ is not differentiable at the origin by arguing that the suitability inequality for $h$ can still be transformed into the suitability inequality for $f$ because in both cases the left-hand side is $0$.
The utility of \Cref{suitnorm} is demonstrated in \Cref{theorem:x-alpha-norm}, which extends the results of \tempdel{Bhawalkar et al.\ }\cite{bhawalkar2013coevolutionary} for the PoA of $\norm{x}^\alpha$ to our more general model. This will be discussed in the following subsection

\subsection{The Exact Worst-Case PoA}
\label{subsection:exact-poa}

The goal of this section is to show how we can get an expression for the worst-case PoA of a large class of functions by complementing the upper bounds provided by the tools in subsection \ref{subsection:poa-tools} with matching lower bounds. There are two main tools for establishing these lower bounds. One is \Cref{lowerboundthm}, which provides a general worst-case PoA lower bound of $\frac 2 {e\ln 2}$ and, as we will see below, is tight for a large class of functions. The other is the previously described \Cref{scalarlowerboundlemma}, which shows that the worst-case PoA of a function is at least the infimum of $\frac \lambda \kappa$ over positive $\lambda, \kappa$ such that said function is $(\lambda, \kappa)$-suitable. This lemma was used in the proof of \Cref{lowerboundthm}, but can also be used directly to determine the worst-case PoA of certain functions.

\paragraph{Applications of \Cref{lowerboundthm}}
\Cref{lowerboundthm} most immediately allows us to give the exact worst-case PoA for the following large class of functions derived from $e^x$:
\begin{corollary}\label{expcorol}
    For any $g : \R^m \to \R$ which is nonconstant, differentiable, and convex, the worst-case PoA of $e^{g(z)}$ is exactly $\frac 2 {e\ln 2}$.
\end{corollary}
\begin{proof}
    The lower bound follows directly from \Cref{lowerboundthm}. To see the upper bound, first note that by \Cref{expthm}, there exist $\lambda, \kappa > 0$ such that $\frac \lambda \kappa = \frac 2 {e\ln 2}$ and $e^x$ is $(\lambda, \kappa)$-suitable. Therefore, by \Cref{suitconvex}, $e^{g(x)}$ is also $(\lambda, \kappa)$-suitable, so by \Cref{suitcorol} the worst case PoA of $e^{g(x)}$ is at most $\frac \lambda \kappa = \frac 2 {e\ln 2}$.
\end{proof}
We further have the following interesting example of a function for which we can state the exact worst-case PoA.
\begin{corollary}\label{cosh}
    The worst-case PoA of $\cosh$ is $\frac 2 {e\ln 2}$.
\end{corollary}
\begin{proof}
    The lower bound follows directly from \Cref{lowerboundthm}. To see the upper bound, recall that $\cosh x = \frac {e^x + e^{-x}} 2$. We have as part of the proof of \Cref{expthm} that there exist $\lambda, \kappa > 0$ such that $\frac \lambda \kappa = \frac 2 {e\ln 2}$ and $e^x$ is $(\lambda, \kappa)$-suitable; by \Cref{suitlinear} (applying a linear transformation to the input of a function preserves suitability), $e^{-x}$ is also $(\lambda, \kappa)$-suitable. Therefore, by \Cref{suitcombo} (a nonnegative linear combination of suitable functions is also suitable), $\cosh x = \frac {e^x + e^{-x}} 2$ is also $(\lambda, \kappa)$-suitable. Thus, by \Cref{suitcorol} we have that the worst-case PoA is upper bounded by $\frac \lambda \kappa = \frac 2 {e\ln 2}$.
\end{proof}
We can extend \Cref{cosh} to functions of the form $\cosh g(z)$ for convex and differentiable $g : \R^m \to \R$ using \Cref{suitconvex}, as well as to functions of the form $\cosh \norm{z}$ for norms $\norm{\cdot} : \R^m \to \R$ using \Cref{suitnorm}. 

The above techniques can also be reapplied more generally to show that for any function of the form $f(x) = \sum_{k = 1}^\infty w_k e^{b_k x^{a_k}}$ where $a_k, b_k \in \R$ and $w_k \geq 0$, the worst-case PoA of $f$ is exactly $\frac 2 {e\ln 2}$ (unless $f$ is constant).

We note the below corollary in order to make an interesting observation:
\begin{corollary}\label{corollary:h-shifted-exact-wcpoa}
    For all functions $h$ whose worst-case PoA is exactly $\frac 2 {e\ln 2}$, the worst-case PoA of $h(z) + c$ where $c > 0$ is also $\frac 2 {e\ln 2}$.
\end{corollary}
\begin{proof}
    The lower bound still follows directly from \Cref{lowerboundthm}. The upper bound applies because increasing a cost function by a constant increases both the social cost of the Nash equilibrium and the optimal social cost by the same amount, which can only decrease their ratio, which is the PoA.
\end{proof}
However, the same is not true of \textit{decreasing} the cost function, i.e. considering $h(x) - c$ for some $c > 0$. In particular, note that the first terms of the Taylor series of $\cosh x$ are $1 + \frac {x^2} 2$. Therefore, consider subtracting $1$ from $\cosh x$: the resulting function will behave like $\frac {x^2} 2$ close to $x = 0$, and if we consider an example network that causes the function $x^2$ to attain the PoA of $\frac 9 8$ shown by \tempdel{Bindel et al.\ }\cite{bindel2015bad}, then this example network should approach the same PoA as scale it down so that the effects of higher order terms vanish. Thus, the PoA of $\cosh x - 1$ should be $\frac 9 8$. It may be interesting to consider how the PoA of $\cosh x - c$ varies for $c \in [0, 1]$.
\paragraph{Direct applications of \Cref{scalarlowerboundlemma}} We first demonstrate an example of how \Cref{scalarlowerboundlemma} can be directly applied by considering another important result of \tempdel{Bhawalkar et al.}\cite{bhawalkar2013coevolutionary} specifically for functions of the form $|x|^\alpha$ in the \munagalamodel model. This result is expressed in the following lemma.
\begin{lemma}\label{bhawalkarthm}
    Define
    $\zeta(\alpha)$ as in \Cref{eq:zetaintro}. For $\alpha > 1$, over $\lambda, \kappa > 0$ such that the function $|x|^\alpha : \R \to \R$ is $(\lambda, \kappa)$-suitable, the minimum value of $\frac \lambda \kappa$ is $\zeta(\alpha)$.
\end{lemma}

\begin{corollary}\label{bhawalkarcorol}
    For $\alpha > 1$, the worst-case PoA of $|x|^\alpha : \R \to \R$ is exactly $\zeta(\alpha)$.
\end{corollary}
\begin{proof}
    By \Cref{bhawalkarthm}, over $\lambda, \kappa > 0$ such that $|x|^\alpha$ is $(\lambda, \kappa)$-suitable, the minimum value of $\frac \lambda \kappa$ is $\zeta(\alpha)$. Therefore, by \Cref{scalarlowerboundlemma}, the worst-case PoA of $|x|^\alpha$ is lower bounded by $\zeta(\alpha)$. Conversely, given such $\lambda, \kappa$, we can apply \Cref{suitcorol} to get that the worst-case PoA of $|x|^\alpha$ is upper bounded by $\zeta(\alpha)$.
\end{proof}

We will now proceed to prove a theorem showing the exact worst-case PoA for functions of the form $f(\norm{z})$ where $\norm{\cdot}$ is a vector norm and $f$ is a scalar function. As a corollary, we will show \Cref{theorem:x-alpha-norm} which extends \tempdel{Bhawalkar et al.}\cite{bhawalkar2013coevolutionary}'s result to the functions $\norm{z}^\alpha : \R^m \to \R$.
\begin{theorem}\label{normthm}
    Let $u : \R_{\ge 0} \to \R$ be a function which is nonnegative, convex, and differentiable, and satisfies $u'(0) = 0$. Let $\norm{z}$ be any norm on $\R^m$ which is differentiable everywhere other than the origin. Define $h : \R^m \to \R$ by $h(z) = u(\norm{z})$. Let $f : \R \to \R$ be the even extension of $u$, so that $f(x) = u(\norm{x})$. Then the worst-case PoA of $h$ is equal to the worst-case PoA of $f$.
\end{theorem}
\begin{proof}
    We first apply \Cref{lowerboundthm} to see that we can choose pairs $(\lambda, \kappa)$ such that $f$ is $(\lambda, \kappa)$-suitable and $\frac \lambda \kappa$ is arbitrarily close to the worst-case PoA of $f$.

    Given such pairs, we can then apply \Cref{suitnorm} with $g = \norm{\cdot}$. The requirements are met, as $\norm{\cdot}$ is a norm, while $f$ is nonnegative and satisfies $f'(0) = 0$ because $u$ satisfies both constraints. Thus, for pairs $(\lambda, \kappa)$ such that $\frac \lambda \kappa$ approaches the worst-case PoA of $u$, the worst-case PoA of $h$ is upper bounded by $\frac \lambda \kappa$. It follows that the worst-case PoA of $h$ is at most the worst-case PoA of $f$.

    We can then show that the worst-case PoA of $h$ is at least the worst-case PoA of $f$ using \Cref{restrictionlemma}. In fact, \textit{any} restriction of $h$ is equal to a function $u(\norm{v}\norm{x})$ by the definition of a norm; this is then equal to $f(\norm{v}x)$, meaning that by \Cref{restrictionlemma} the worst-case PoA of $h$ is at least the worst-case PoA of any such function $x \mapsto f(\norm{v}x)$. The scaling of the argument by $\norm{v}$ is equivalent to scaling the argument by including $\norm{v}$ in $A_{ij}, B_{ij}$,  meaning that the definition of the worst-case PoA for any function $f(\norm{v}x)$ is automatically equivalent to the definition of the worst-case PoA for $f$. Thus, the worst-case PoA of $h$ must be at least the worst-case PoA of $f$, completing the proof.
\end{proof}

As a corollary of this theorem, we have the result stated in \Cref{theorem:x-alpha-norm} from the introduction: by combining \Cref{normthm} with \Cref{bhawalkarcorol} we get that the worst-case PoA of $\norm{z}^\alpha$ is exactly $\zeta(\alpha)$, extending \tempdel{Bhawalkar et al.}\cite{bhawalkar2013coevolutionary}'s result  to our \vectormunagalamodel model. Furthermore the worst-case PoA of $\norm{z}^\alpha$ for \textit{any} differentiable norm $\norm{z}$ is exactly $\zeta(\alpha)$, extending their result in an even more general sense. For example, we could take $\norm{z}$ to be the $p$-norm for $1 < p < \infty$ (i.e. $\norm{\cdot}_p$).

In the case of $\alpha = 2$, we can say that we have extended \tempdel{Bindel et al.}\cite{bindel2015bad}'s tight $\frac 9 8$ PoA bound to not only the model mentioned before with cost functions of the form $z^T M z$, but any cost function $\norm{Mz}^2$. In fact, because we show $(\lambda, \kappa)$-suitability of $\norm{Mz}^2$ using \Cref{suitnorm} for the same pair $(\lambda, \kappa)$ for all norms $\norm{\cdot}$, we can even consider a model where not only $M$ but $\norm{\cdot}$ is allowed to depend on the persons it corresponds to, and the $\frac 9 8$ PoA bound will continue to hold as a result of \Cref{suitabilitythm}.

\section{The Clique Model}
\label{section:clique-model}

As mentioned in the introduction, a notable application of our general model is that we can reduce the \cliquevectormodel model, where the network is divided into cliques that use best-response dynamics for the entire clique in order to determine their expressed opinions, to said general model. In fact, our general model is powerful enough that we can reduce a clique version of the model to it. We reintroduce the clique model below; this definition can be applied on top of any of the models defined in the introduction, but we will take it to be a modification of our general model for the purposes of achieving the most general results.
\begin{definition}\label{cliquegeneralmodel}
    A \emph{clique opinion formation game} is a modification of an opinion formation game in which the $n$ persons are partitioned into $t$ cliques $C_1 \sqcup \dotsb \sqcup C_t = \{1, \ldots, n\}$. We define a cost for the entire clique $q_i(z) = \sum_{j \in C_i} c_i(z)$; we then modify the update so that in each clique $C_i$, persons $j \in C_i$ collaborate to choose expressed opinions so that the clique cost $C_i$ is minimized under the assumption that the expressed opinions of members of other cliques remain the same.

    The Nash equilibrium $x$ is then the set of expressed opinions which is invariant under this update, i.e. $x$ satisfies $\nabla_j q_i(z) = 0$ for all $j \in C_i$. $y$ remains unchanged; the PoA is thus redefined as $\frac {SC(x)} {SC(y)}$ for this new $y$.
\end{definition}
We then have the following theorem in parallel to \Cref{suitabilitythm}; the major difference here is that the cost functions $f_{ij}$ must be $(\lambda, \kappa, 1)$-suitable in addition to $(\lambda, \kappa, 2)$-suitable, i.e. $(\lambda, \kappa)$-suitable. The reason for this is that when $i, j$ are in the same clique, $f_{ij}$ now acts as an internal cost function. This is not a significant restriction, as the tools demonstrated in section \ref{section:tools} for proving PoA upper bounds generally assume $(\lambda, \kappa)$-suitability.
\begin{theorem}
    \label{cliquethm}
    Consider a symmetric \textit{clique} opinion formation game whose cost functions $f_{ij}, g_i$ are nonnegative and differentiable. If for some $\lambda, \kappa > 0$, we have that all $f_{ij}$ are both $(\lambda, \kappa, 2)$-suitable and $(\lambda, \kappa, 1)$-suitable, and all $g_i$ are $(\lambda, \kappa, 1)$-suitable, then the price of anarchy of the opinion formation game is at most $\frac \lambda \kappa$.
\end{theorem}
As mentioned before, the proof proceeds by a reduction to \Cref{suitabilitythm}. The general idea of this reduction is to create a person corresponding to each clique whose expressed opinion is a stack vector of the expressed opinions of the members of the clique, then choose linear transformations to apply to the inputs of cost functions so that the expressed opinions of individual members can be extracted and used as inputs to the original cost functions. The proof is given below.

\begin{proof}[Proof of \Cref{cliquethm}]
    For any symmetric clique opinion formation game, we will demonstrate a normal (not clique) symmetric opinion formation game whose dynamics are identical. For brevity, we will refer to these as the clique game and the normal game respectively. Furthermore, we will differentiate internal opinions and cost functions defined for the new game by notating them as, for example, $s_i', g_i'$.

   For each clique $C_i$ in our clique game, say that we have $|C_i|$ persons $j_{i, 1}, \ldots, j_{i, |C_i|} \in C_i$. We will define a person $i$ in the normal game who is meant to represent the entire clique $C_i$. The dimension $d_i'$ of their expressed opinion will be the sum $d_{j_{i, 1}} + \dotsb + d_{j_{i, |C_i|}}$ of the dimensions of expressed opinions of the members of $C_i$ in the clique game. We will then consider the expressed opinion $z_i'$ to be equal to the \textit{stack vector} $\begin{pmatrix}
       z_{j_{i, 1}} \dotsb z_{j_{i, |C_i|}}
   \end{pmatrix}^T$; in other words, the first $d_{j_{i, 1}}$ entries of $z_i'$ are the entries of $z_{j_{i, 1}}$, the next $d_{j_{i, 2}}$ entries of $z_i'$ are the entries of $z_{j_{i, 2}}$, and so on so that a given $z$ corresponds to both a set of expressed opinions in the clique game and a set of expressed opinions in the normal game.
   
   Given this, we aim to choose the cost functions and internal opinions of the normal game so that the cost $c_i'(z)$ in the normal game is equal to the cost $q_i(z)$ in the clique game. It will then follow that the Nash equilibrium $x$ and social optimum $y$ are the same in the normal game as in the clique game and have the same social costs, implying that the PoA of the two games is the same.

   We will simply define $R_i' = I$ and $s_i' = 0$ so that the penalty for person $i$ due to internal opinion in the normal game is simply the $g_i'(z_i')$ which we have yet to define. $g_i'$ will now be defined so as to encapsulate all costs internal to the clique $C_i$.

   First note that for any clique $C_i$ and person $j \in C_i$, the entries of $z_j$ appear in $z_i'$, meaning that we can define a linear transformation $L_j : \R^{d_i'} \to {\R^{d_j}}$ such that $L_jz_i' = z_j$. Given this, we can express the cost $g_j(R_j z_j - s_j)$ as $g_j(R_j L_j z_i' - s_j)$. We note that given that $g_j$ is $(\lambda, \kappa, 1)$-suitable, $a \mapsto g_j(R_j L_j a - s_j)$ is as well; this is an instance of the fact that applying an affine transformation to the input of a function preserves $(\lambda, \kappa, p)$-suitability, as shown in \Cref{suitlinear}.

   Similarly, for $j, k \in C_i$, we can express the cost $f_{jk}(A_{jk}z_j + B_{jk}z_k)$ as $f_{jk}(A_{jk}L_jz_i' + B_{jk}L_k z_i') = f_{jk}((A_{jk}L_j + B_{jk}L_k)z_i')$. Just as before, because $f_{jk}$ is $(\lambda, \kappa, 1)$-suitable, the function $a \mapsto f_{jk}((A_{jk}L_j + B_{jk}L_k)a)$ is as well by \Cref{suitlinear}.

   We then set:
   \begin{align*}
       g_i'(z_i') = \sum_{j \in C_i} g_j(R_j L_j z_i' - s_j) + \sum_{j, k \in C_i} f_{jk}((A_{jk}L_j + B_{jk}L_k)z_i').
   \end{align*}
   $g_i'$ is a nonnegative linear combination of $(\lambda, \kappa, 1)$-suitable functions; as a result, it is $(\lambda, \kappa, 1)$-suitable as well by \Cref{suitcombo}. It therefore meets the constraints of \Cref{suitabilitythm}. Furthermore, by the definition of $L_j$, we have:
   \begin{align}\label{cliquegineq}
       g_i'(z_i') = \sum_{j \in C_i} g_j(R_j z_j - s_j) + \sum_{j, k \in C_i} f_{jk}(A_{jk}z_j + B_{jk}z_k),
   \end{align}
   so that $g_i'$ accounts for the portion of $q_i$ which is internal to the clique $C_i$, i.e. the penalties due to internal opinions and the penalties between two members of the clique.

   We now define $f_{il}'$ for two persons $i, l$ in the normal game. We will define $A_{il}', B_{il}'$ by defining the images $A_{il}'z_i', B_{il}'z_l'$, which we will then define procedurally as follows. For each pair $(j, k)$ of persons $j, k$ in the clique game such that $j \in C_i$ and $k \in C_l$, given the penalty $f_{jk}(A_{jk}z_j + B_{jk}z_k)$ corresponding to the difference in opinion between persons $j, k$, we will append the entries in $A_{jk}L_j z_i'$ to $A_{il}'z_i'$ and append the entries in $B_{jk}L_k z_l'$ (which must have the same dimension) to $B_{il}'z_l'$. In this way, we can define a linear transformation $M_{jk}$ that picks out those entries, so that $M_{jk}A_{il}'z_i' = A_{jk}L_j z_i' = A_{jk}z_j$, and $M_{jk}B_{il}'z_l' = B_{jk}L_k z_l' = B_{jk}z_k$. This then allows us to define $f_{il}'$ as follows:
   \begin{align*}
       f_{il}'(a) = \sum_{j \in C_i} \sum_{k \in C_l} f_{jk}(M_{jk}a).
   \end{align*}
   This has two implications. First, because all $f_{jk}$ are $(\lambda, \kappa, 2)$-suitable, the functions $a \mapsto f_{jk}(M_{jk}a)$ are $(\lambda, \kappa, 2)$-suitable as well by \Cref{suitlinear}, and so $f_{il}'$ itself is as well by \Cref{suitcombo}. Combined with the fact that we naturally have $A_{li}' = B_{il}'$ and $B_{li}' = A_{il}$ by the symmetry of the above definitions, we meet the constraints of \Cref{suitabilitythm}.

   Second, in the context of the normal game, we see that:
   \begin{align*}
       f_{il}'(A_{il}'z_i' + B_{il}'z_l')
       &= \sum_{j \in C_i} \sum_{k \in C_l} f_{jk}(M_{jk}(A_{il}'z_i' + B_{il}'z_l')) &\EqComment{Definition of $f_{il}'$}\\
       &= \sum_{j \in C_i} \sum_{k \in C_l} f_{jk}(A_{jk}z_j + B_{jk}z_k), &\EqComment{Definition of $M_{jk}$}
   \end{align*}
   showing that our defined cost between $i, l$ accounts for the portion of $q_i$ due to penalties with persons outside of the clique.

   We therefore have that $g_i'(z_i')$ together with the defined costs $f_{il}'(A_{il}'z_i' + B_{il}'z_l')$ are equivalent to the cost $q_i(z)$; we formally show this below:

    {\allowdisplaybreaks
   \begin{align*}
       c_i'(z)
       &= g_i'(R_i'z_i' - s_i') + \sum_{l \neq i}f_{il}'(A_{il}'z_i' + B_{il}'z_l')&\EqComment{Definition of $c_i$}\\
       &= g_i'(z_i') + \sum_{l \neq i}f_{il}'(A_{il}'z_i' + B_{il}'z_l')&\EqComment{Choice of $R_i', s_i'$}\\
       &= \sum_{j \in C_i} g_j(R_j z_j - s_j) + \sum_{j, k \in C_i} f_{jk}(A_{jk}z_j + B_{jk}z_k) \\ &\qquad \quad + \sum_{l \neq i}f_{il}'(A_{il}'z_i' + B_{il}'z_l')&\EqComment{\Cref{cliquegineq}}\\
       &= \sum_{j \in C_i} g_j(R_j z_j - s_j) + \sum_{j, k \in C_i} f_{jk}(A_{jk}z_j + B_{jk}z_k) \\ &\qquad \quad  + \sum_{l \neq i}\sum_{j \in C_i} \sum_{k \in C_l} f_{jk}(A_{jk}z_j + B_{jk}z_k)&\EqComment{Preceding identity for $f_{il}'$}\\
       &= \sum_{j \in C_i}\left[g_j(R_j z_j - s_j) + \sum_{k \neq j} f_{jk}(A_{jk}z_j + B_{jk}z_k)\right]&\EqComment{Rearranging sums}\\
       &= \sum_{j \in C_i} c_j(z)&\EqComment{Definition of $c_j$}\\
       &= q_i(z).&\EqComment{Definition of $q_i$}\\
\end{align*}
}

   Thus, $c_i'(z) = q_i(z)$ as desired, and so the PoA of the normal game is the same as that of the clique game. We can then conclude by \Cref{suitabilitythm} that the PoA of the normal game is at most $\frac \lambda \kappa$, meaning that the PoA of the clique game is also at most $\frac \lambda \kappa$.
\end{proof}

As an extension, we can prove an analogue of \Cref{suitcorol} for clique opinion formation games, allowing us to apply the tools of Section \ref{section:tools}.
\begin{corollary}{\label{cliquecorol}}
    Suppose that a function $h : \R^m \to \R$ is nonnegative and differentiable. If for some $\lambda, \kappa > 0$, $h$ is $(\lambda, \kappa)$-\emph{suitable}, then the PoA of any symmetric clique opinion formation game whose cost functions are of the form $ch$ is bounded by $\frac \lambda \kappa$.
\end{corollary}
\begin{proof}
    By \Cref{suitlinear}, all cost functions $f_{ij}, g_i$ of any such clique opinion formation game are $(\lambda, \kappa)$-suitable, and so the $f_{ij}$ are both $(\lambda, \kappa, 2)$-suitable and $(\lambda, \kappa, 1)$-suitable and the $g_i$ are $(\lambda, \kappa, 1)$-suitable, so by \Cref{cliquethm} the PoA of such a game is at most $\frac \lambda \kappa$.
\end{proof}

\section*{Acknowledgements}
    The work is partially supported by DARPA QuICC, ONR MURI 2024 award on Algorithms, Learning, and Game Theory, Army-Research Laboratory (ARL) grant W911NF2410052, NSF AF:Small grants 2218678, 2114269, 2347322.

\newpage
\bibliographystyle{alpha}
\bibliography{references}

\appendix
\newpage
\section{Omitted Proofs}

\subsection{Local Smoothness}
\label{appendix:localsmoothness-proof}
We restate \Cref{localsmoothness}, which gives in our notation a result of Roughgarden and Schoppmann \cite{roughgarden2015local}, then provide a proof.
\begin{lemmarestate}{\ref{localsmoothness}}
    Consider a general cost-minimization game where we have $n$ players, each of whom can take some action represented by a vector $z_i \in \R^m$ (defining $z = (z_1, \ldots, z_n)$ as before), for which they have some penalty $c_i(z)$, such that the social cost is $SC(z) = \sum_i c_i(z)$. Assume that $SC(z)$ is always nonnegative.

    If for all $n$-tuples $x, y$ such that $x_i, y_i \in \R^n$, we have:
    \begin{align*}
        \sum_i \nabla_ic_i(x)^T(y_i - x_i) \leq \lambda SC(y) - \kappa SC(x),
    \end{align*}
    then the PoA of the cost minimization game is at most $\frac \lambda \kappa$.
\end{lemmarestate}
\begin{proof}
    First, because $x$ is a Nash equilibrium, it must be that for each person $i$, their personal cost $c_i(x)$ cannot be decreased by changing $x_i$: therefore, $\nabla_i c_i(x) = 0$.
    It then follows that:
    \begin{align*}
        \lambda SC(y) - \kappa SC(x)
         & \geq \sum_i \nabla_ic_i(x)^T(y_i - x_i)
         & \EqComment{Local $(\lambda, \kappa)$ smoothness}
        \\
         & \geq \sum_i 0
         & \EqComment{see above}
        \\
         & = 0.
    \end{align*}
    We can rearrange $\lambda SC(y) - \kappa SC(x)$ to get $\lambda SC(y) \geq \kappa SC(x)$, which then implies that $\frac {SC(x)} {SC(y)} \leq \frac \lambda \kappa$.

    Finally, taking $y$ to be the optimum, for any Nash equilibrium $x$ we have that the ratio of its social cost to the optimal cost is at most $\frac \lambda \kappa$, so the PoA is at most $\frac \lambda \kappa$ as desired.
\end{proof}

\subsection{Lower Bound on PoA}
\label{section:poa-lowerbound-inf}

The following lemma gives a result of \tempdel{Bhawalkar et al.\ }\cite{bhawalkar2013coevolutionary} that is slightly modified for our model. Specifically, the \munagalamodel model required that the cost functions be symmetric, and their proofs used this constraint. However, we can get the same results for our model \textit{without} the requirement that the cost functions be symmetric using the same proofs with some technical modifications. For example, this is important to show that the worst-case PoA of $e^x$ is \textit{exactly} $\frac 2 {e\ln 2}$.
\begin{lemmarestate}{\ref{scalarlowerboundlemma}}
    Given a nonnegative, differentiable, and convex function $h : \R \to \R$, the worst-case PoA of $h$ is lower bounded by the infimum of $\frac \lambda \kappa$ for which $\lambda, \kappa > 0$ and $h$ is $(\lambda, \kappa)$-suitable.
\end{lemmarestate}
\tempdel{Bhawalkar et al.}\cite{bhawalkar2013coevolutionary}'s original result, which is Theorem 3.5 in their paper, is proven in two lemmas. We give them below with the necessary modifications. In the below lemma, which is a modification of Lemma 3.6 in their paper, we briefly describe a part of their proof that is completely unchanged; we then show an elegant technique to prove the remainder in a simpler manner that does not rely on $h$ being symmetric.
\begin{lemma} \label{appendixlb1}
    Given a nonnegative and differentiable function $h : \R \to \R$, say that we have a finite number of pairs of constraints of the form:
    \begin{align*}
        h'(x)\frac {y - x} 1 &\leq \lambda h(y) - \kappa h(x),\\
        h'(x)\frac {y - x} 2 &\leq \lambda h(y) - \kappa h(x),
    \end{align*}
    for $x, y \in \R$. Let $\gamma$ be the minimum value of $\frac \lambda \kappa$ among $\lambda, \kappa > 0$ satisfying said constraints, and fix $\lambda, \kappa$ correspondingly. If $\gamma > 1$, then there exist $x_1, y_1, x_2, y_2 \in \R$ such that:
    \begin{align*}
        h'(x_1)\frac {y_1 - x_1} 2 &= \lambda h(y_1) - \kappa h(x_1),\\
        h'(x_2)(y_2 - x_2) &= \lambda h(y_2) - \kappa h(x_2),
    \end{align*}
    and $h'(x_1)(x_1 - y_1), h'(x_2)(x_2 - y_2)$ have opposite signs.
    \begin{proof}
        First note that \tempdel{Bhawalkar et al.}\cite{bhawalkar2013coevolutionary}'s notation uses $\mu$ in place of $\kappa$, so that $1 - \mu = \kappa$.

        We briefly summarize part of \tempdel{Bhawalkar et al.}\cite{bhawalkar2013coevolutionary}'s proof: they note that only considering constraints for one value of $p$, $\gamma$ will always be at most $1$: because $h$ is convex, it is automatically $(\lambda, \kappa, p)$-suitable for $\lambda, \kappa = p$. Therefore, given that $\gamma > 1$, it must be attained as the intersection of a constraint for $p = 1$ and a constraint for $p = 2$.

        The remainder of their proof shows that $x_1 - y_1, x_2 - y_2$ have opposite signs. This part uses the constraint that $h$ is symmetric to derive that $h$ is increasing for $x > 0$ (given that $h$ is convex); they then argue based on the fact that $h'$ is positive.

        We can instead argue as follows. When $h'(x)(y - x)$ is positive, the constraint $h'(x)(y - x) \leq \lambda h(y) - \kappa h(x)$ is at least as tight as $h'(x)\frac {y - x} 2 \leq \lambda h(y) - \kappa h(x)$ because $h'(x)(y - x) > h'(x)\frac {y - x} 2$. Similarly, when $h'(x)(y - x)$ is negative, the constraint $h'(x)\frac {y - x} 2 \leq \lambda h(y) - \kappa h(x)$ is at least as tight as $h'(x)(y - x) < \lambda h(y) - \kappa h(x)$.
        
        Therefore, considering the constraints of the form $h'(x)\frac {y - x} 1 \leq \lambda h(y) - \kappa h(x)$ only when $h'(x)(y - x)$ is positive and the constraints of the form $h'(x)\frac {y - x} 2 \leq \lambda h(y) - \kappa h(x)$ only when $h'(x)(y - x)$ is negative will yield the same $\gamma$, so we will necessarily have that $h'(x_1)(y_1 - x_1)$ is negative and $h'(x_2)(y_2 - x_2)$ is positive, meaning that they have opposite signs.
    \end{proof}
\end{lemma}
The next lemma is a modified version of Lemma 3.7 from \tempdel{Bhawalkar et al.\ }\cite{bhawalkar2013coevolutionary}. We simply reprove the entire lemma as multiple details change in order to prove the desired result, and because the proof given by \tempdel{Bhawalkar et al.\ }\cite{bhawalkar2013coevolutionary} leaves out many details, meaning that the following proof also serves expository purposes.
\begin{lemma}\label{appendixlb2}
    We are given a differentiable function $h : \R \to \R$. If for some $\lambda, \kappa > 0$ and $x_1, y_1, x_2, y_2 \in S$ we have that:
    \begin{align*}
        h'(x_1)\frac {y_1 - x_1} 2 &= \lambda h(y_1) - \kappa h(x_1),\\
        h'(x_2)(y_2 - x_2) &= \lambda h(y_2) - \kappa h(x_2),
    \end{align*}
    and that $h'(x_1)(x_1 - y_1), h'(x_2)(x_2 - y_2)$ have opposite signs, then the worst-case PoA of $h$ is at least $\frac \lambda \kappa$.
\end{lemma}
\begin{proof}
We first define $a = \frac {y_1 x_2 - x_1 y_2} {x_2 - y_2}$ and $b = \frac {y_1 x_2 - x_1 y_2} {x_1 - y_1}$. We will define a symmetric opinion formation game whose cost functions are $h$ dilated by $a, b$; their forms have been carefully chosen so that the PoA can be made $\frac \lambda \kappa$ with a simple choice of internal opinions. Importantly, note that because the numerator in their definitions are the same and the denominators in their definitions are $x_2 - y_2$ and $x_1 - y_1$, which have opposite signs, we have that $a, b$ also have opposite signs.

We now describe the specific opinion formation game whose PoA we will show to be $\frac \lambda \kappa$. This game will have three persons, who we will refer to as $\pneg, \pzero, \ppos$. These names are explanatory -- persons $\ppos$ and $\pneg$ will be reflections of each other in the sense that their internal opinions will be opposite and their cost functions will be the same, and person $\pzero$ will have no internal penalty, meaning that by symmetry person $\pzero$ will have expressed opinion $0$ in both the Nash equilibrium and the social optimum, and so the behavior of persons $\pneg, \ppos$ will always be opposites of each other.

We thus choose cost functions $f_{\pneg0} = f_{\ppos0} = wh$ for a $w \geq 0$ to be chosen later, $g_\pneg = g_\ppos = h$, and $g_\pzero = 0$. Therefore, all cost functions are $h$ multiplied by a nonnegative weight, so this game is included in the definition of the worst-case PoA of $h$, meaning that the worst-case PoA of $h$ is at least the PoA of this game.

We then choose $A_{\pneg0} = B_{0\pneg} = -a, B_{\pneg0} = A_{0\pneg} = a$, so that the cost due to difference of opinion between $\pneg$ and $\pzero$ is $wh(az_\pzero - az_\pneg)$. Similarly, we choose $A_{\ppos0} = B_{0\ppos} = a, B_{\ppos0} = A_{0\ppos} = -a$, so that the cost due to difference of opinion between $\ppos$ and $\pzero$ is $wh(az_\ppos - az_\pzero)$. We further choose $R_\pneg = -1, s_\pneg = 1$, so that the internal cost for person $\pneg$ is $h(-bz_\pneg - b)$, which can be written as $h(-b - bz_\pneg)$. Similarly, we choose $R_\ppos = 1, s_\ppos = 1$ so that the internal cost for person $\ppos$ is $h(bz_\ppos - b)$. The choices of $R_\pzero, s_\pzero$ are irrelevant as $g_\pzero = 0$.

Therefore, the costs for each person are as follows:
\begin{align*}
    c_\ppos(z) &= h(bz_\ppos - b) + wh(az_\ppos - az_\pzero),\\
    c_\pzero(z) &= wh(az_\pzero - az_\pneg) + wh(az_\ppos - az_\pzero),\\
    c_\pneg(z) &= h(-b - bz_\pneg) + wh(az_\pzero - az_\pneg).
\end{align*}
The added symmetry in this game is key. For convenience, we will call a set of expressed opinions $z$ \textit{symmetric} if we have $z_\pzero = 0$ and $z_\pneg = -z_\ppos$. For a symmetric $z$, the cost $wh(z_\pzero - z_\pneg)$ will be identical to the cost $wh(z_\ppos - z_\pzero)$, because $z_\pzero - z_\pneg = -z_\pneg = z_\ppos = z_\ppos - z_\pzero$. Similarly, the cost $h(-1 - z_\pneg)$ will be identical to the cost $h(z_\ppos - 1)$. Crucially, this symmetry does not rely on $h$ itself being symmetric.

As a consequence, if for a symmetric $z$, we have that person $\ppos$'s choice $z_\ppos$ is personally optimal given the other persons' expressed opinions, we will have that the same is true for person $\pneg$ by symmetry. Furthermore, for \textit{any} symmetric $z$ the partial derivative of $c_\pzero(z)$ with respect to $z_\pzero$ must be $0$ by symmetry, meaning that the choice $z_\pzero$ is personally optimal for any symmetric $z$. Therefore, to show that some symmetric $x$ is a Nash equilibrium, it suffices to show that person $\ppos$'s choice is personally optimal.

A second consequence is that the social cost for a symmetric $z$ becomes equal to:
\begin{align}\label{appendixssc}
SC(z) = 2(h(bz_\ppos - b) + 2wh(az_\ppos)),
\end{align}
because both $c_\ppos(z)$ and $c_\pneg(z)$ are equal to $h(bz_\ppos - b) + wh(az_\ppos - az_\pzero)$, $c_\pzero(z)$ is equal to $wh(az_\ppos - az_\pzero) + wh(az_\ppos - az_\pzero)$, and $z_\pzero = 0$. Both of these consequences will be used below.

We now define $x$ so that $x_\ppos = \frac {x_1} a, x_\pzero = 0, x_\pneg = -\frac {x_1} a$, and $y$ so that $y_\ppos = \frac {y_1} a, y_\pzero = 0, y_\pneg = -\frac {y_1} a$. Note that both $x, y$ are symmetric. $x, y$ are in fact the Nash equilibrium and social optimum respectively. We will show below that $x$ is the Nash equilibrium; however, we will not show that $y$ is the social optimum as in order to lower bound the PoA of the game we only need $y$ as an example.

As a preliminary step, we perform two calculations for later use. We first show that $x_\ppos - 1 = \frac {x_2} b$:
\begin{align*}
    x_\ppos - 1
    &= \frac {x_1} a - 1&\EqComment{Definition of $x_\ppos$}\\
    &= \frac {x_1 - a} a &\EqComment{Writing as single fraction}\\
    &= \frac {x_1 - \frac {y_1x_2 - x_1 y_2} {x_2 - y_2}} {\frac {y_1x_2 - x_1 y_2} {x_2 - y_2}} &\EqComment{Definition of $a$}\\
    &= \frac {x_1(x_2 - y_2) - (y_1x_2 - x_1 y_2)} {y_1x_2 - x_1 y_2} &\EqComment{Multiplying by denominator}\\
    &= \frac {x_1 - y_1} {y_1 x_2 - x_1 y_2} x_2 &\EqComment{Simplification}\\
    &= \frac {x_2} b. &\EqComment{Definition of $b$}
\end{align*}
We next show that $y_\ppos - 1 = \frac {y_2} b$, which proceeds along similar lines:
\begin{align*}
    y_\ppos - 1
    &= \frac {y_1} a - 1 &\EqComment{Definition of $y_\ppos$}\\
    &= \frac {y_1 - a} a &\EqComment{Writing as single fraction}\\
    &= \frac {y_1 - \frac {y_1 x_2 - x_1 y_2} {x_2 - y_2}} {\frac {y_1 x_2 - x_1 y_2} {x_2 - y_2}} &\EqComment{Definition of $a$}\\
    &= \frac {y_1(x_2 - y_2) - (y_1 x_2 - x_1 y_2)} {y_1 x_2 - x_1 y_2}&\EqComment{Multiplying by denominator}\\
    &= \frac {x_1 - y_1} {y_1 x_2 - x_1 y_2} y_2 &\EqComment{Simplification}\\
    &= \frac {y_2} b. &\EqComment{Definition of $b$}
\end{align*}
We will now show that $x_\ppos$ is the optimal choice for person $\ppos$ given that the other persons' opinions remain fixed; as stated earlier, this is sufficient to show that $x$ is a Nash equilibrium given that $x$ is symmetric.

More precisely, we will choose $w$ so that $x_\ppos$ is personally optimal. The necessary constraint on $w$ comes from setting $\frac \partial {\partial z_\ppos}|_{z = x} c_\ppos(z)$ to $0$, so we calculate said partial derivative here:
\begin{align*}
    \frac \partial {\partial z_\ppos}|_{z = x} c_\ppos(z)
    &= \frac \partial {\partial z_\ppos}|_{z = x}[h(bz_\ppos - b) + wh(az_\ppos - az_\pzero)] &\EqComment{Expanding $c_\ppos(z)$}\\
    &= bh'(bx_\ppos - bs_\ppos) + wah(ax_\ppos - x_\pzero) &\EqComment{Chain rule}\\
    &= bh'\left(b \frac {x_2} b\right) wah(ax_\ppos - x_\pzero) &\EqComment{$x_\ppos - s_\ppos = \frac {x_2} b$}\\
    &= bh'\left(b \frac {x_2} b\right) wah(a\frac {x_1} a) &\EqComment{Definition of $x$}\\
    &= bh'(x_2) + wah'(x_1).
\end{align*}
We therefore choose $w = -\frac {bh'(x_2)} {ah'(x_1)}$ in order to make $x_\ppos$ personally optimal. We can obtain a useful expression for $w$ as follows:
\begin{align*}
    w
    &= -\frac {bh'(x_2)} {ah'(x_1)} &\EqComment{Definition of $w$}\\
    &= -\frac {\frac {y_1 x_2 - x_1 y_2} {x_1 - y_1}} {\frac {y_1 x_2 - x_1 y_2} {x_2 - y_2}} \cdot \frac {h'(x_2)} {h'(x_1)} &\EqComment{Definitions of $a, b$}\\
    &= -\frac {x_2 - y_2} {x_1 - y_1} \cdot \frac {h'(x_2)} {h'(x_1)}. &\EqComment{Cancellation}
\end{align*}
Importantly, $w$ must be nonnegative. To see this, recall that $h'(x_1)(y_1 - x_1)$ and $h'(x_2)(y_2 - x_2)$ have opposite signs; $w$ is simply the negative of the ratio of the latter to the former, and so it must be positive.

We now finally show that $\frac {SC(x)} {SC(y)} = \frac \lambda \kappa$. It is equivalent to show that $\lambda SC(y) - \kappa SC(x) = 0$. We will actually show that $\frac {\lambda SC(y) - \kappa SC(x)} 2 = 0$ as this is also equivalent and removes a redundant factor of $2$ from the argument:
\begin{align*}
    &\frac {\lambda SC(y) - \kappa SC(x)} 2 \\
    &= \frac {\lambda \cdot 2(h(by_\ppos - b) + 2wh(ay_\ppos)) - \kappa \cdot 2(h(bx_\ppos - b) + 2wh(ax_\ppos))} 2 &\EqComment{\Cref{appendixssc}}\\
    &= \lambda (h(by_\ppos - b) + 2wh(ay_\ppos)) - \kappa (h(bx_\ppos - b) + 2wh(ax_\ppos)) &\EqComment{Cancellation}\\
    &= \lambda \left(h(by_\ppos - b) + 2wh\left(a\frac {y_1} a\right)\right) - \kappa \left(h(bx_\ppos - b) + 2wh\left(a\frac {x_1} a\right)\right) &\EqComment{Definitions of $x, y$}\\
    &= \lambda \left(h(by_\ppos - b) + 2wh\left(a\frac {y_1} a\right)\right) - \kappa \left(h\left(b\frac {x_2} b\right) + 2wh\left(a\frac {x_1} a\right)\right) &\EqComment{$x_\ppos - 1 = \frac {x_2} b$}\\
    &= \lambda \left(h\left(b\frac {y_2} b\right) + 2wh\left(a\frac {y_1} a\right)\right) - \kappa \left(h\left(b\frac {x_2} b\right) + 2wh\left(a\frac {x_1} a\right)\right) &\EqComment{$y_\ppos - 1 = \frac {y_2} b$}\\
    &= \lambda (h(y_2) + 2wh(y_1)) - \kappa (h(x_2) + 2wh(x_1)) &\EqComment{Cancellations}\\
    &= (\lambda h(y_2) - \kappa h(x_2)) + 2w(\lambda h(y_1) - \kappa h(x_1)) &\EqComment{Rearranging}\\
    &= h'(x_2)(y_2 - x_2) + 2w h'(x_1)\frac {y_1 - x_1} 2 &\EqComment{Statement of lemma}\\
    &= h'(x_2)(y_2 - x_2) + wh'(x_1)(y_1 - x_1) &\EqComment{Cancellation}\\
    &= h'(x_2)(y_2 - x_2) - \frac {x_2 - y_2} {x_1 - y_1} \cdot \frac {h'(x_2)} {h'(x_1)} h'(x_1)(y_1 - x_1) &\EqComment{Expanding $w$}\\
    &=h'(x_2)(y_2 - x_2) - h'(x_2)(y_2 - x_2) &\EqComment{Cancellation}\\
    &= 0.
\end{align*}
We thus have that $\frac {SC(x)} {SC(y)} = \frac \lambda \kappa$. Therefore, let $y'$ be a social optimum. Then $SC(y') \leq SC(y)$, so the PoA of the game is $\frac {SC(x)} {SC(y')} \geq \frac {SC(x)} {SC(y)} = \frac \lambda \kappa$. It follows that the worst-case PoA of $h$ is at least $\frac \lambda \kappa$.
\end{proof}
We now prove \Cref{scalarlowerboundlemma}; this proof is essentially a restatement of an argument of \tempdel{Bhawalkar et al.\ }\cite{bhawalkar2013coevolutionary} and is included for completeness.
\begin{proof}{Proof of \Cref{scalarlowerboundlemma}}
    Let $I$ be the infimum of $\frac \lambda \kappa$ for which $\lambda, \kappa > 0$ and $h$ is $(\lambda, \kappa)$-suitable on $S$ so that we wish to show that the worst-case PoA of $h$ is at least $I$. By the definition of suitability, $I$ is the infimum of $\frac \lambda \kappa$ over $\lambda, \kappa > 0$ satisfying the pair of constraints:
    \begin{align*}
        h'(x)\frac {y - x} 1 &\leq \lambda h(y) - \kappa h(x),\\
        h'(x)\frac {y - x} 2 &\leq \lambda h(y) - \kappa h(x),
    \end{align*}
    for all $x, y \in S$. As the rational numbers are dense in the real numbers, it is equivalent to consider the above pair of constraints for only rational $x, y$. The number of pairs of constraints is then countable, meaning that we can order them, and define $I_n$ to be the infimum of $\frac \lambda \kappa$ over $\lambda, \kappa > 0$ satisfying the first $n$ pairs of constraints. It then follows that $\lim_{n \to \infty} I_n = I$, so that it suffices to show that the worst-case PoA of $h$ is at least $I_n$ for all $n$.

    For any $n$, first note that if $I_n \leq 1$ then we are done, as the PoA is always at least $1$, and so the worst-case PoA is as well. Otherwise, we can apply \Cref{appendixlb1} with $\gamma = I_n$ to get that there exist $x_1, y_1, x_2, y_2 \in S$ and $\lambda, \kappa > 0$ such that: \begin{align*}
        h'(x_1)\frac {y_1 - x_1} 2 &= \lambda h(y_1) - \kappa h(x_1),\\
        h'(x_2)(y_2 - x_2) &= \lambda h(y_2) - \kappa h(x_2),
    \end{align*}
    and $x_1 - y_1, x_2 - y_2$ have opposite signs. We can then apply \Cref{appendixlb2} to get that the worst-case PoA of $h$ is at least $\frac \lambda \kappa$, completing the proof.
\end{proof}

\subsection{Equivalence of The \arbitrarysymmetricmodelcapital Model and The \ourmodelcapital Model} \label{section:equivalence_arb_het}

In this section, we prove that the \arbitrarysymmetricmodel model and the \ourmodel model are in fact equivalent, and consequently, the \ourmodel model is the most general model with symmetric cost functions.

First recall that in the \arbitrarysymmetricmodel model, the cost function for person $i$ is defined as:

$$
    c_i(z) = g_i(z) + \sum_{j \neq i}{f_{ij}(z_i,z_j)},
$$
where we have the symmetry condition $f_{ij}(x,y) = f_{ij}(y,x)$ for each $i$ and $j$. In the \ourmodel model, the cost function takes the following form:

$$
c_i(z) = g_i(R_iz_i - s_i) + \sum_{j \neq i}{f_{ij}(A_{ij}z_i + B_{ij}z_j)},
$$

with the symmetry condition $f_{ij}(A_{ij}z_i + B_{ij}z_j)=f_{ji}(A_{ji}z_j + B_{ji}z_i)$.

It is straightforward to see that the \ourmodel model can be reduced to the \arbitrarysymmetricmodel model since the cost function allows for more general functions. We proceed by showing that the converse is also true.

To do this, let us represent the cost functions and internal opinions in the \arbitrarysymmetricmodel model using the superscript "arb". Each cost function $f(x, y)$ where $f_{ij}^{\text{arb}} : \R^{m_i} \times \R^{m_j} \to \R$ can be viewed as a function $f : \R^{m_i + m_j} \to \R$ on the concatenation of $x, y$ -- that is, if we let $z \in \R^{m_i + m_j}$ be a vector where the first $m_i$ entries in $z$ are the entries in $x$ and the last $m_j$ entries in $z$ are the entries in $y$, then $f_{ij}(z) = f_{ij}^{\text{arb}}(x, y)$. Therefore, if for $i < j$ we let $A_{ij} = B_{ji} : \R^{m_i} \to \R^{m_i + m_j}$ be a linear map from $x$ to $x$ appended with $m_j$ entries equal to $0$, and let $B_{ij} = A_{ji}: \R^{m_j} \to \R^{m_i + m_j}$ be a linear map from $y$ to $y$ prepended with $m_i$ entries equal to $0$, then we find that $f_{ij}(A_{ij}x + B_{ij}y) = f_{ij}(z) = f_{ij}^{\text{arb}}(x, y)$. We can then finally take $g_i = g_i^{\text{arb}}$ with $R_i$ and $s_i$ equal to $0$ to find that all costs are equivalent.

\subsection{Lemmas \ref{cneglemma} and \ref{cposlemma}}
\label{appendix:cposneglemma-proof}
We will first restate \Cref{cneglemma} and then provide its proof.
\begin{lemmarestate}{\ref{cneglemma}}
For any function $h$ : $[M, \infty)$ (for some $M \in \R$) which is continuous, positive, and satisfies $\int_M^{\infty} h(x)dx = \infty$, we have that for any $\epsilon > 0$, there exist $z \in [M + 1, \infty)$ such that
\begin{align*}
    \frac {\int_z^{z - 1} h(x)dx} {h(z)} \geq -1 - \epsilon.
\end{align*}
\end{lemmarestate}
\begin{proof}
    We will show this using a proof by contradiction. Suppose that there exists some $\epsilon > 0$ such that for all $z \in [M + 1, \infty)$, we have $\frac {\int_z^{z - 1} h(x)dx} {h(z)} < -1 - \epsilon$. 
    Then $\int_z^{z - 1} h(x)dx < (-1 - \epsilon) h(z)$, and so by flipping the signs we have $\int_{z-1}^z h(x)dx > (1 + \epsilon) h(z)$.
    For any $T > M+1$, taking the integral of this inequality over $[M+1,T]$ leads to
    \begin{align}
        \label{inequality:integral-c-minusone}
        \int_{M+1}^T \int_{z-1}^z h(x)dx dz > (1 + \epsilon) \int_{M+1}^T h(z)dz.
    \end{align}
    We can then upper bound the left hand side as follows:
    \begin{align*}
        \int_{M+1}^T \int_{z-1}^z h(x)dx dz = \int_M^{M+1} h(z)(z-M)dz + \int_{M +1}^{T - 1} h(z)dz + \int_{T-1}^T h(z)(T - z)dz
        \leq \int_M^T h(z)dz.
    \end{align*}
    Combining this with Inequality \ref{inequality:integral-c-minusone} gives us $\int_M^T h(z)dz > (1 + \epsilon) \int_{M+1}^T h(z)dz$. Crucially, we can split the left hand side integral into $\int_M^{M + 1}h(z)dz + \int_{M+1}^T h(z)dz$, which then allows us to cancel the term $\int_{M+1}^T h(z)dz$ from the right hand side, leading to the inequality
    \begin{align*}
        \int_M^{M + 1}h(z)dz > \epsilon \int_{M+1}^T h(z)dz.
    \end{align*}
    Taking the limit as $T \to \infty$, the left hand side remains constant and finite (since $h(z)$ is bounded from both sides for $z \in [M, M+1]$) while the right hand side should go to infinity as per the assumptions of the lemma, leading to a contradiction and concluding the proof.
\end{proof}

We now restate \Cref{cposlemma} and then prove it.
\begin{lemmarestate}{\ref{cposlemma}}
For any function $h$ : $[M, \infty)$ (for some $M \in \R$) which is continuous, positive, and satisfies $\int_M^{\infty} h(x)dx = \infty$, we have that for any $\epsilon > 0$, there exists $z \in [M, \infty)$ such that:
\begin{align*}
    \frac {\int_z^{z + 1} h(x)dx} {h(z)} \geq 1 - \epsilon.
\end{align*}
\end{lemmarestate}
\begin{proof}
    We again proceed by contradiction. Suppose that there exists some $\epsilon > 0$ such that for all $z \in [M, \infty)$. We then have $\frac {\int_z^{z + 1} h(x)dx} {h(z)} < 1 - \epsilon$, and so $\int_z^{z+1} h(x)dx < (1 - \epsilon) h(z)$. For any $T > M+1$, taking the integral of this inequality over $[M, T]$ results in
    \begin{align}
        \label{inequality:intergral-c-one}
        \int_M^T \int_z^{z+1} h(x)dxdz < (1 - \epsilon) \int_{M}^{T} h(z)dz.
    \end{align}
    We can then lower bound the left hand side as follows:
    \begin{align*}
        \int_M^T \int_z^{z+1} h(x)dxdz = \int_M^{M+1} h(z)(z-M)dz + \int_{M+1}^T h(z)dz + \int_T^{T + 1}h(z)(T + 1 - z)dz 
        \ge \int_{M+1}^T h(z)dz.
    \end{align*}
    Combining this with Inequality \ref{inequality:intergral-c-one} we have $\int_{M+1}^T h(z)dz < (1 - \epsilon) \int_{M}^{T} h(z)dz$. Similarly to before, we can split the right hand side into $(1 - \epsilon) \int_M^{M + 1} h(z)dz + (1 - \epsilon)\int_{M + 1}^T h(z)dz$; the second term can then be cancelled with the left hand side, leading to the inequality
    \begin{align*}
        \epsilon\int_{M+1}^T h(z)dz < (1 - \epsilon) \int_{M}^{M + 1} h(z)dz.
    \end{align*}
    Taking the limit as $T \to \infty$, one can see that the right hand side remains constant and finite (since $h(z)$ is bounded from both sides for $z \in [M, M+1]$), while the left hand side should go to infinity as per the assumptions of the lemma, leading to a contradiction and concluding the proof.
\end{proof}

\section{Convergence}
\label{appendix:convergence}
Consider the \quadraticvectormodel model's cost function from \Cref{eq:quadratic-vectorized:cost}, restated below.
\begin{align*}
    c_i(z) = (z_i - s_i)^T R_i (z_i - s_i) + \sum_j (z_i - z_j)^T W_{ij} (z_i - z_j).
\end{align*}
As stated before, we assume that $W_{ij}$ is positive semidefinite for every $i,j$, and $R_i$ is also positive semidefinite for every $i$. In this section, we consider the case where all $R_i$ matrices are further positive semidefinite. We derive an update rule from the above cost function by taking $\nabla_i c_i(z) = \frac{\partial c_i(z)}{\partial z_i}$,
\begin{align}
    \label{equation:quadratic-vectorized:derivative}
    \nabla_i c_i(z) = 2R_iz_i - 2R_is_i + 2 \left(\sum_j W_{ij}\right)z_i - 2 \sum_j W_{ij}z_j,
\end{align}
because $R_i$ and $W_{ij}$ are all symmetric. To arrive at the update rule in \Cref{eq:quadratic-vectorized:update}, we put $\nabla_i c_i(z) = 0$ and get:
\begin{align*}
    z_i(t+1) = \left(R_i + \sum_j W_{ij}\right)\inv \left(R_is_i + \sum_j W_{ij} z_j(t)\right).
\end{align*}
Because $R_i$ is positive definite and each $W_{ij}$ is positive semidefinite, $R_i + \sum_j W_{ij}$ is positive definite and thus invertible. To go further we need some definitions. We call a matrix
$$
    A = \begin{bmatrix}
        A_{11} & A_{12} & \cdots & A_{1n} \\
        A_{21} & A_{22} & \cdots & A_{2n} \\
        \vdots & \vdots & \ddots & \vdots \\
        A_{n1} & A_{n2} & \cdots & A_{nn}
    \end{bmatrix}\in \R^{nm \times nm}
$$
a \textit{block} matrix, where $A_{ij} \in \R^{m \times m}$. We further call $A$, \textit{block diagonal}, if $A_{ij} = \mathbf{0}_m$ for every $i \neq j$, where $\mathbf{0}_m$ is a $m \times m$ matrix of zeros. We also call a vector $v = \begin{pmatrix} v_1^T, \dots, v_n^T \end{pmatrix}^T \in \R^{nm}$, a \textit{block} vector, where $v_i \in \R^m$. For a matrix $A$, let $\rho(A)$ be its \emph{spectral radius}.

Let $W$ be a block matrix, where $W_{ij}$ is the weight matrix on the edge $(i, j)$, which we have been using so far. Let $\Lambda$ be a block diagonal matrix, where $\Lambda_{ii} = \left(R_i + \sum_j W_{ij}\right)\inv$. Then we can rewrite the update rule in \Cref{eq:quadratic-vectorized:update} as:
\begin{align}
    z_i(t+1) = \Lambda_{ii} \sum_j W_{ij} z_j(t) + \Lambda_{ii} R_i s_i.
\end{align}
Let $z(t) = \begin{pmatrix} z_1(t)^T, \dots, z_n(t)^T \end{pmatrix}^T$ be the block vector of all opinions at time step $t$, and let $R$ be a block diagonal matrix having each $R_i$ on its diagonal. Then if $s = \begin{pmatrix} s_1^T, \dots, s_n^T \end{pmatrix}^T$ is the block vector of all internal opinions, we can rewrite the above update rule in a single equation as:
\begin{align}
    \label{eq:quadratic-vectorized:short-form}
    z(t+1) = \Lambda W z(t) + \Lambda R s.
\end{align}
We now state a lemma from Ravazzi, Frasca, Tempo, and Ishii \cite{ravazzi2014ergodic}, which helps us understand the convergence of the above update rule.
\begin{lemma}
    \cite[Proposition 1]{ravazzi2014ergodic}
    \label{lemma:convergenet-ravazzi}
    The model:
    $$
        x(t + 1) = Ax(t) + s,
    $$
    is convergent to 
    $$
        x^* = (I - A)\inv s,
    $$
    for any $x(0)$, if and only if $\rho(A) < 1$.
\end{lemma}
By applying the above lemma to the update rule in \Cref{eq:quadratic-vectorized:short-form}, we get the following corollary.
\begin{corollary}
    \label{corollary:quadratic-vectorized:convergence-schur-stable}
    The update rule in \Cref{eq:quadratic-vectorized:short-form} is convergent for any $z(0)$, if and only of $\rho(\Lambda W) < 1$.
\end{corollary}
Next, we will introduce the \textit{weight-normalized} property.
\begin{definition}[Weight-Normalized Instance]
    An instance $G = (V, E)$ with $n$ nodes is normalized, if $\sum_j W_{ij} = I$ for every node $i$.
\end{definition}

Next, we will introduce the \textit{clone transformation}, which transforms an instance into a normalized instance.
\begin{definition}[Clone Transformation]
    Consider an instance $G = (V, E)$ with $n$ nodes. The clone transformation of $G$ is the graph $H$ created as follows:
    \begin{itemize}
        \item Put two copies of $G$ in $H$, numbering the nodes such that the first copy's nodes are from $1$ to $n$ (and so the second copy's nodes are from $n + 1$ to $2n$).
        \item Choose a positive real $d$ big enough, so that $dI - \sum_{j=1}^{n} W_{ij}$ is positive definite, for all $i$.
        \item For every node $i$ in the first copy, add an undirected edge between $i$ and $i + n$ (the corresponding node in the second copy) with the weight matrix $dI - \sum_{j=1}^{n} W_{ij}$.
        \item Divide every $W_{ij}$ and $R_i$ by $d$ (element wise).
    \end{itemize}
\end{definition}
It is trivial to see that the clone transformation of an instance, is weight-normalized. Note that the graph can be directed or undirected for the above clone transformation. Also, note that the clone transformation is not unique; any $d$ big enough leads to a different weight-normalized instance.

Before going further, we have the following lemma discussing the uniqueness of the Nash equilibrium.
\begin{lemma}
    Let $G = (V, E)$ be an undirected instance with $n$ nodes, having each $R_i$ as positive definite. The Nash equilibrium of $G$ is unique.
\end{lemma}
\begin{proof}
    One can notice that at a Nash equilibrium $x$, for every $i$, $\nabla_i c_i(x) = 0$. This means that:
    \begin{align*}
        R_i x_i - R_i s_i + \left(\sum_j W_{ij}\right) x_i - \sum_j W_{ij} x_j = 0,
    \end{align*}
    which by rearranging we have that:
    \begin{align*}
        R_i x_i + \sum_j W_{ij} x_i - \sum_j W_{ij} x_j = R_i s_i.
    \end{align*}
    Let $L$ be a block matrix where $L_{ij} = -W_{ij}$ for $i \neq j$, and $L_{ii} = \sum_j W_{ij}$ for every $i$. Let $R$ be a block diagonal matrix having each $R_i$ on its diagonal. Then the above equation for every $i$ can be written into one equation as $Rx + Lx = Rs$, where $x$ is the stacked vector of the Nash equilibrium and $s$ is the stacked vector of all internal opinions. Note that the block matrix $R$ is positive definite since every block on the diagonal is positive definite. Also $L$ is positive semidefinite. We show this by showing that $v^T L v \ge 0$ for every $v \in \R^{nm}$. Precisely, let $v = (v_1^T, \dots, v_n^T)^T$ where $v_i \in \R^m$, then we have:
    \begin{align*}
        v^T L v &= \sum_{i,j} v_i^T L_{ij} v_j &\EqComment{Expansion} \\
        &= \sum_i v_i^T \left(\sum_j W_{ij}\right) v_i - \sum_{i \neq j} v_i^T W_{ij} v_j &\EqComment{Definition of $L$} \\
        &= \sum_{i,j} v_i^T W_{ij} v_i - v_i^T W_{ij} v_j. &\EqComment{Bringing the first term into the second sum}
    \end{align*}
    Now if we write every term in the equation above twice, we get:
    \begin{align*}
        2v^T L v &= \sum_{i,j} 2v_i^T W_{ij} v_i - 2v_i^T W_{ij} v_j &\EqComment{See above} \\
        &= \sum_{i,j} v_i^T W_{ij} v_i - 2v_i^T W_{ij} v_j + v_j^T W_{ij} v_j &\EqComment{Brining a copy of the term for $j$ to the sum for $i$} \\
        &= \sum_{i,j} (v_i - v_j)^T W_{ij} (v_i - v_j), &\EqComment{Symmetry of $W_{ij}$ and grouping the terms} \\
        &\ge 0, &\EqComment{$W_{ij}$ is positive semidefinite}
    \end{align*}
    and thus $L$ is positive semidefinite.

    Now going back to the equation $Rx + Lx = Rs$, we have that $(R + L)x = Rs$, where $R$ is positive definite and $L$ is positive semidefinite. Therefore $R + L$ is positive definite and thus invertible. Hence the unique Nash equilibrium is at $x = (R + L)\inv Rs$.
\end{proof}

We then have the following lemma, regarding the Nash equilibrium of the clone transformation.
\begin{lemma}
    Let $G = (V, E)$ be an undirected instance with $n$ nodes, having each $R_i$ as positive definite. Let $H$ be its clone transformation. If $x_1, \dots, x_n$ is the unique Nash equilibrium for $G$, $x_1, \dots x_n, x_1, \dots, x_n$ is the unique Nash equilibrium for $H$.
\end{lemma}
\begin{proof}
    We first notice that for any instance, at the Nash equilibrium $x$, $\nabla_i c_i(x) = 0$ for any $i$, because otherwise, the node $i$ can change its opinion in the direction of $\nabla_i c_i(x)$ to reduce its penalty. The other way is true too; if $\nabla_i c_i(x) = 0$ for all $i$, then $x$ is a Nash equilibrium.
    
    Now let $c_i^G(z)$ and $c_i^H(z)$ be the cost functions for node $i$ in the graphs $G$ and $H$ respectively. As stated, $\nabla_i c_i^G(x^G) = 0$ where $x^G = (x_1, \dots, x_n)$, the sequence of all expressed opinions at the Nash equilibrium. This means that:
    \begin{align*}
        \nabla_i c_i^G(x^G) = 2R_i x_i - 2R_i s_i + 2 \left(\sum_j W_{ij}\right)x_i - 2 \sum_j W_{ij} x_j = 0.
    \end{align*}
    Then considering $\nabla_i c_i^H(x^H)$ for $1 \le i \le n$, where $x^H = (x_1, \dots, x_n, x_1, \dots, x_n)$, we have that:
    \begin{align*}
        \nabla_i c_i^H(x^H) = \frac{1}{d} \left(2R_i x_i - 2R_i s_i + 2 \left(\sum_j W_{ij} + W_{i,i+n}\right)x_i - 2\sum_j W_{ij} x_j - 2 W_{i,i+n} x_{i + n} \right),
    \end{align*}
    where $d$ is the real number chosen in the clone transformation. Since $x_i = x_{i + n}$, the additional terms having $W_{i,i+n}$ cancel out from the equation above and hence $\nabla_i c_i^H(x^H) = \frac{1}{d} c_i^G(x^G) = 0$, showing that $x^H$ is a Nash equilibrium for $H$, concluding the proof.
\end{proof}

Using the lemma above, it is enough to only consider weight-normalized instances for our purposes. This is because every other instance can be transformed into a weight-normalized instance while keeping the Nash equilibrium.

We then have the following theorem, regarding the convergence of undirected weight-normalized instances.

\begin{theorem}
    \label{theorem:quadratic-vectorized:convergence}
    Let $G = (V, E)$ be an undirected instance (i.e. $W_{ij} = W_{ji}$ for every $i,j$) and have every $R_i$ as positive definite. If $G$ is weight-normalized, then the update rule in \Cref{eq:quadratic-vectorized:short-form} is convergent for any $x(0)$.
\end{theorem}
\begin{proof}
    We prove the theorem, by applying \Cref{corollary:quadratic-vectorized:convergence-schur-stable}. So we only need to show that $\rho(\Lambda W) < 1$. As a property of spectral radius, $\rho(A) \le \norm{A}$ for any matrix $A$, where $\norm{\cdot}$ is any vector-induced norm. Hence $\rho(\Lambda W) \le \norm{\Lambda W}_2 \le \norm{\Lambda}_2 \cdot \norm{W}_2$ using the sub-multiplicative property of the 2-norm. To prove the statement, we show that $\norm{\Lambda}_2 < 1$ and $\norm{W}_2 \le 1$. Also note that $\norm{A}_2 = \sigma_{\max}(A)$, where $\sigma_{\max}(A)$ is the maximum signular value of $A$. 

    Let us first analyze $\Lambda$. We have that $\Lambda_{ii} = \left(R_i + \sum_j W_{ij}\right)\inv = \left(R_i + I\right)\inv$, because $G$ is weight-normalized. It is trivial to see that if $0 < \lambda_1 \le \dots \le \lambda_m$ are eigenvalues of $R_i$, then $0 < \frac{1}{\lambda_m + 1} \le \dots \le \frac{1}{\lambda_1 + 1} < 1$ are eigenvalues of $\left(R_i + I\right)\inv$. Because $R_i$ is positive definite, so is $\Lambda_{ii}$, and so $\Lambda$ is a block diagonal matrix where each block is positive definite, and thus $\Lambda$ itself is positive definite. For any positive definite matrix, its eigenvalues are the same as its singular values, hence $\norm{\Lambda}_2 = \lambda_{\max}(\Lambda)$ where $\lambda_{\max}(\Lambda)$ is the maximum eigenvalue of $\Lambda$. Again because $\Lambda$ is block diagonal, its eigenvalues are the union of eigenvalues of all $\Lambda_{ii}$. We have already shown that eigenvalues of each $\Lambda_{ii}$ are all less than 1, hence $\norm{\Lambda}_2 < 1$, as intended.

    Next we show that $\norm{W}_2 \le 1$. First notice that $W$ is a symmetric matrix, because $W_{ij} = W_{ji}$ and also each $W_{ij}$ is symmetric itself, thus $W^T = W$. So the singular values of $W$, are the absolute values of its eigenvalues. We show that if $\lambda$ is an eigenvalue of $W$, then $|\lambda| \le 1$ and thus $\rho(W) \le 1$.
    To do this, we show that the matrices $I + W$ and $I - W$ are positive semidefinite. If $\lambda$ is an eigenvalue of $W$, then $1 + \lambda$ is an eigenvalue of $I + W$, thus if $I + W$ is positive semidefinite we have that $1 + \lambda \ge 0$. Following the same argument for $I - W$, we have that $1 - \lambda \ge 0$ and overall we have that $|\lambda| \le 1$. So all that remains is to show that $I + W$ and $I - W$ are positive semidefinite.

    Let us start by $I + W$. We have to show that $v^T (I + W) v \ge 0$ for any $v \in \R^{nm}$. Considering $v$ as a block vector, let $v_1, \dots, v_n$ be its blocks where $v_i \in \R^m$. Then we have:
    \begin{align*}
        v^T (I + W) v &= v^T v + v^T W v \\
        &= \left(\sum_{i} v_i^T v_i \right) + \left(\sum_{i,j} v_i^T W_{ij} v_j\right) &\EqComment{By expansion} \\
        &= \sum_{i,j} v_i^T W_{ij} v_i + v_i^T W_{ij} v_j. &\EqComment{$G$ is weight-normalized}
    \end{align*}
    Let us add each term twice in the equation above, and then divide the whole summation by 2. When readding the terms $v_i^T W_{ij} v_i$, rewrite them to $v_j^T W_{ij} v_j$ which is equal in total to $v_i^T W_{ij} v_i$. When reading the terms $v_i^T W_{ij} v_j$, take the transpose of the term and write $v_j^T W_{ij}^T v_i$ and then simplify to $v_j^T W_{ij} v_i$, because $W_{ij}$ is symmetric. So we have:
    \begin{align*}
        v^T (I + W) v &= \frac{1}{2} \sum_{i,j} v_i^T W_{ij} v_i + v_i^T W_{ij} v_i + v_i^T W_{ij} v_j + v_i^T W_{ij} v_j &\EqComment{Adding each term twice}\\
        &= \frac{1}{2} \sum_{i,j} v_i^T W_{ij} v_i + v_j^T W_{ij} v_j + v_i^T W_{ij} v_j + v_j^T W_{ij} v_i &\EqComment{See above} \\
        &= \frac{1}{2} \sum_{i,j} (v_i + v_j)^T W_{ij} (v_i + v_j). &\EqComment{Grouping the terms}
    \end{align*}
    Now because each $W_{ij}$ is positive semidefinite, $(v_i + v_j)^T W_{ij} (v_i + v_j) \ge 0$ and so $v^T (I + W) v \ge 0$, as intended.

    A similar argument holds for $I - W$. Precisely:
    \begin{align*}
        v^T (I - W) v &= v^T v - v^T W v \\
        &= \left(\sum_{i} v_i^T v_i \right) - \left(\sum_{i,j} v_i^T W_{ij} v_j\right) &\EqComment{By expansion} \\
        &= \sum_{i,j} v_i^T W_{ij} v_i - v_i^T W_{ij} v_j &\EqComment{$G$ is weight-normalized} \\
        &= \frac{1}{2} \sum_{i,j} v_i^T W_{ij} v_i + v_i^T W_{ij} v_i - v_i^T W_{ij} v_j - v_i^T W_{ij} v_j &\EqComment{Adding each term twice} \\
        &= \frac{1}{2} \sum_{i,j} v_i^T W_{ij} v_i + v_j^T W_{ij} v_j - v_i^T W_{ij} v_j - v_j^T W_{ij} v_i &\EqComment{See above} \\
        &= \frac{1}{2} \sum_{i,j} (v_i - v_j)^T W_{ij} (v_i - v_j), &\EqComment{Grouping the terms}
    \end{align*}
    which shows that $v^T (I - W) v \ge 0$, completing the proof.
\end{proof}

\section{An Example of Unbounded PoA with Non-convex Cost Functions}

\label{section:non-convex_unbounded_example}

We will demonstrate here an example of an opinion formation game whose cost functions are nonnegative and continuous (in fact, polynomial) but not convex, such that the PoA of said game is unbounded, which is to say that its optimal social cost is $0$ while there exists a Nash equilibrium whose social cost is positive. This example explains why we demand that cost functions be convex when showing PoA bounds. Note that the following example involves one-dimensional opinions, meaning that even in the one-dimensional case, convexity is required in order to attain PoA bounds.

For simplicity of presentation, we will define the opinion formation game in the \arbitrarysymmetricmodel model. Recall that the \arbitrarysymmetricmodel model can be reduced to our \ourmodel model, and in particular the cost function $f_{12}$ below can be expressed in the \ourmodel model by setting $A_{12} = \begin{pmatrix}
    1 \\ 0
\end{pmatrix}$, $B_{12} = \begin{pmatrix}
    0 \\ 1
\end{pmatrix}$, and $f_{12}(\begin{pmatrix}
    x \\ y
\end{pmatrix}) = (1 - x)^2y^2 + x^2(1 - y)^2$.

\begin{example}\label{example:nonconvex}
    Given $\epsilon > 0$, define an opinion formation game as follows: the game has two players who each have one dimensional opinions with the cost functions $g_k(z_k) = \epsilon z_k^2$ and $f_{12}(z_1, z_2) = (1 - z_1)^2z_2^2 + z_1^2(1 - z_2)^2$.
\end{example}
In \Cref{example:nonconvex}, the cost function $f_{12}$ is not convex, leading to the PoA being unbounded as shown in the following theorem.

\begin{theorem}\label{theorem:nonconvex_poa_infinite}
    \Cref{example:nonconvex} with $\epsilon = \frac 1 8$ gives a one-dimensional opinion formation game whose cost functions are polynomial and nonnegative and which has unbounded PoA.
\end{theorem}
\begin{proof}
    First note that $g_k$ is a square and $f_{12}$ is a sum of two squares, so both are nonnegative.
    
    Further note that by setting $y$ to have $y_1 = y_2 = 0$, all cost functions become $0$, which then must be the optimal social cost, as the social cost must be nonnegative.

    Now set $x$ to have $x_1 = x_2 = t$, with $t = \frac 3 4$. We claim that $x$ is a Nash equilibrium of the opinion formation game. To see this, by symmetry, it suffices to show that $x_1$ cannot be changed to some $z_1 \neq t$ so as to decrease $c_1(z)$, i.e. with $x_2$ fixed, $x_1$ is a global minimum of $c_1(z)$. Observe that with $z_2$ fixed at $t$, $c_1(z) = g_1(z_1) + f_{12}(z_1, t) = \epsilon z_1^2 + (1 - z_1)^2t^2 + z_1^2(1 - t)^2$ is a sum of squares and so is convex in $z_1$, meaning that it suffices to show that $\frac {\partial {c_1(z)}} {\partial z_1}|_{z_1 = t} = 0$, which we see below:
    \begin{align*}
        \frac {\partial {c_1(z)}} {\partial z_1}|_{z_1 = t}
        &= \frac {\partial } {\partial z_1}|_{z_1 = t} \left[\epsilon z_1^2 + (1 - z_1)^2t^2 + z_1^2(1 - t)^2\right]\\
        &= \left[2\epsilon z_1 - 2(1 - z_1)t^2 + 2z_1(1 - t)^2\right]_{z_1 = t}\\
        &= 2\left(\epsilon t - (1 - t)t^2 + t(1 - t)^2\right)\\
        &= 2\left(\epsilon t - t^2 + t^3 + t - 2t^2 + t^3\right)\\
        &= 2\left(\epsilon t + 2t^3 - 3t^2 + t\right)\\
        &= 2\left(\frac 1 8 \cdot \frac 3 4 + 2\left(\frac 3 4\right)^3 - 3\left(\frac 3 4\right)^2 + \frac 3 4\right)\\
        &= 2\left(\frac 3 {32} + \frac {27} {32} - \frac {27} {16} + \frac 3 4\right)\\
        &= 2\frac {3 + 27 - 54 + 24} {32}\\
        &= 0.
    \end{align*}
    Therefore, $x$ is a Nash equilibrium. However, $g_1(x_1) = g_2(x_2) = \epsilon t^2 > 0$, and $f_{12}$, so $x$ must have a positive social cost, which, because the optimal social cost is $0$, gives an unbounded PoA.
\end{proof}

\section{The \quadraticvectormodelcapital Model Without Semidefiniteness}
\label{section:non_positive_definite}

The semidefiniteness assumption in the \quadraticvectormodel model is necessary since without this assumption the game may not have a Nash equilibrium. This is true even in the single-dimensional case (i.e. the \fjmodel model in \Cref{eq:fj:cost}) where semidefiniteness corresponds to nonnegative weights. Below, we provide an example illustrating that without the semidefiniteness condition, a Nash equilibrium may not exist.

\begin{example}\label{example:nondefinite}
    Consider the following instance of the \quadraticvectormodel model as follows: the game involves two players with one-dimensional opinions, where $r_1 = r_2 = 0$, $s_1 = s_2 = 0$ and there is a single undirected edge between them with weight $-1$ (i.e. $w_{12} = w_{21} = -1$).
\end{example}

\begin{lemma}
    Example \ref{example:nondefinite} does not have a Nash equilibrium.
\end{lemma}

\begin{proof}
    We prove this by contradiction. Let $x_1, x_2$ represent a Nash equilibrium. Without loss of generality, assume $x_1 \leq x_2$. Increasing $x_2$ would lower the cost function for node $2$ and decreasing $x_1$ would lower the cost function for node $1$, contradicting with $x_1,x_2$ being a Nash equilibrium. Similarly, one can show that the same game, where $r_1 = r_2 = \epsilon > 0$ does not have a Nash equilibrium either, for any $\epsilon < 1$.
\end{proof}

This scalar game example can be extended to any $m > 1$ dimensions by simply adding dimensions with zero edge weights and internal opinion weights set to $1$.

\end{document}